\documentclass[conference]{IEEEtran}
\title{Custom Hypergraph Categories via Generalized Relations}
\author{\IEEEauthorblockN{Dan Marsden and Fabrizio Genovese}\IEEEauthorblockA{Department of Computer Science, University of Oxford}}

\usepackage{amsmath, amsthm, amsfonts}
\usepackage{tikz}
\usetikzlibrary{positioning,arrows,matrix,calc,intersections,fit, patterns, shapes}

\usepackage{stringdiagrams}
\usepackage{xspace}

\usepackage{thmtools,thm-restate}

\usepackage{color}
\newcommand{\margincomment}[2]{}

\theoremstyle{plain}
\newtheorem{theorem}{Theorem}
\newtheorem{lemma}[theorem]{Lemma}

\newtheorem{proposition}[theorem]{Proposition}
\theoremstyle{definition}
\newtheorem{definition}[theorem]{Definition}
\newtheorem{example}[theorem]{Example}

\theoremstyle{remark}
\newtheorem{remark}[theorem]{Remark}
\numberwithin{theorem}{section}

\newcommand{\catname}[1]{\ensuremath{{\bf #1}}}
\newcommand{\sigmaalg}[2]{\ensuremath{\catname{Alg}(#1,#2)}\xspace}
\newcommand{\sigmaalgt}[3]{\ensuremath{\catname{Alg}^{#1}(#2,#3)}\xspace}
\newcommand{\convexrel}{\catname{ConvexRel}\xspace}
\newcommand{\preord}{\catname{Preord}\xspace}
\newcommand{\cset}{\catname{Set}\xspace}
\newcommand{\pos}{\catname{Pos}\xspace}
\newcommand{\rel}{\ensuremath{\catname{Rel}}\xspace}
\newcommand{\relq}[1]{\ensuremath{\catname{Rel}(#1)}\xspace}
\newcommand{\spanq}[1]{\ensuremath{\catname{Span}(#1)}\xspace}
\newcommand{\spanvq}[2]{\ensuremath{\catname{Span}_{#1}(#2)}\xspace}
\newcommand{\relvq}[2]{\ensuremath{\catname{Rel}_{#1}(#2)}\xspace}
\newcommand{\spanvtq}[3]{\ensuremath{\catname{Span}^{#1}_{#2}(#3)}\xspace}
\newcommand{\relvtq}[3]{\ensuremath{\catname{Rel}^{#1}_{#2}(#3)}\xspace}

\newcommand{\relvqex}[3]{\ensuremath{\catname{Rel}^{\operatorname{#3}}_{#1}(#2)}\xspace}
\newcommand{\linrelvq}[2]{\relvqex{#1}{#2}{lin}}
\newcommand{\affrelvq}[2]{\relvqex{#1}{#2}{aff}}
\newcommand{\relrelvq}[2]{\relvqex{#1}{#2}{rel}}
\newcommand{\cartrelvq}[2]{\relvqex{#1}{#2}{cart}}
\newcommand{\linrelvtq}[3]{\relvqex{#2}{#3}{lin, #1}}

\newcommand{\spanvqex}[3]{\ensuremath{\catname{Span}^{\operatorname{#3}}_{#1}(#2)}\xspace}
\newcommand{\linspanvq}[2]{\spanvqex{#1}{#2}{lin}}
\newcommand{\affspanvq}[2]{\spanvqex{#1}{#2}{aff}}
\newcommand{\relspanvq}[2]{\spanvqex{#1}{#2}{rel}}
\newcommand{\cartspanvq}[2]{\spanvqex{#1}{#2}{cart}}
\newcommand{\linspanvtq}[3]{\spanvqex{#2}{#3}{lin, #1}}

\newcommand{\signature}[1]{\ensuremath{\operatorname{#1}}}

\newcommand{\signaturesex}[1]{\ensuremath{\catname{Sig}^{\operatorname{#1}}}\xspace}
\newcommand{\linsignatures}{\signaturesex{lin}}
\newcommand{\affsignatures}{\signaturesex{aff}}
\newcommand{\relsignatures}{\signaturesex{rel}}
\newcommand{\cartsignatures}{\signaturesex{cart}}

\newcommand{\quantale}[1]{\ensuremath{{\bf #1}}\xspace}

\newcommand{\graph}[1]{\ensuremath{{#1}_\circ}\xspace}
\newcommand{\cograph}[1]{\ensuremath{{}_{\circ}{#1}}\xspace}
\newcommand{\converse}[1]{\ensuremath{{#1}^\circ}\xspace}

\newcommand{\define}[1]{{\bf #1}\xspace}

\allowdisplaybreaks

\begin{document}

\maketitle

\begin{abstract}
  Process theories combine a graphical language for compositional reasoning with an underlying
  categorical semantics. They have been successfully applied to fields such as quantum computation, natural language processing,
  linear dynamical systems and network theory. When investigating a new application, the question arises of
  how to identify a suitable process theoretic model. 

  We present a conceptually motivated parameterized framework for the construction
  of models for process theories.  Our framework generalizes the notion of binary relation
  along four axes of variation, the truth values, a choice of algebraic structure,
  the ambient mathematical universe and the choice of proof relevance or provability.
  The resulting categories are preorder-enriched and provide analogues of relational converse and taking graphs of maps. Our constructions
  are functorial in the parameter choices, establishing mathematical connections between different application domains. 
  We illustrate our techniques by constructing many existing models from the literature,
  and new models that open up ground for further development.
\end{abstract}

\section{Introduction}
The term ``process theory'' has recently been introduced~\cite{CoeckeKissinger2016} to describe compositional theories
of abstract processes. These process theories typically consist of a graphical language for reasoning about composite systems,
and a categorical semantics tailored to the application domain. This compositional perspective has been incredibly
successful in reasoning about questions in quantum computation and quantum foundations. The scope of the process
theoretic perspective encompasses many other application domains, including natural language processing~\cite{CoeckeSadrzadehClark2010},
signal flow graphs~\cite{BonchiSobocinskiZanasi2015}, control theory~\cite{BaezErbele2015},
Markov processes~\cite{BaezFongPollard2016}, electrical circuits~\cite{BaezFong2015} and even linear algebra~\cite{Sobocinski2016}.

When considering a new application of the process theoretic approach, the question arises of how to find a suitable
categorical setting capturing the phenomena of interest. Dagger compact closed categories are of particular importance as they
have an elegant graphical calculus, and many of the examples cited above live in compact closed categories. 

We illustrate the process of constructing new dagger compact closed categories with two examples in the theory of human cognition,
as developed in~\cite{Gardenfors2004, Gardenfors2014}. This is an unconventional application area, and therefore highlights clearly
the challenges faced when trying to model a new problem domain in a process theoretic manner.

As our first example, we consider how the notion of convexity can be incorporated into a compact closed setting. Convexity
is important in mathematical models of cognition, where it is argued that the meaningful concepts should be closed under
forming mixtures. Informally, if we have a space representing animals, then if two points describe dogs, we would expect any points
``in between'' should also be models of the concept of being a dog.

An algebraic model of convexity is given by the Eilenberg-Moore algebras of the finite distribution monad. These algebras, referred to as convex algebras,
are sets equipped with a well behaved operation for forming convex mixtures of elements. Informally, we denote such a convex mixture as
\begin{equation*}
  \sum_i p_i x_i
\end{equation*}
where the~$p_i$ are positive reals summing to one, and the~$x_i$ are elements of the algebra.
This notation is not intended to imply there are independent addition and scaling operations that can be applied to the  individual elements.

The Eilenberg-Moore category of any monad on~\cset, or in fact any regular category, is itself a regular category.
Therefore the category of convex algebras is regular and we can form its category
of relations, denoted~\convexrel. It is well known that the category of relations over a regular category is a dagger compact closed category~\cite{HeunenTull2015}.
Concretely, a convex relation is an ordinary binary relation~$R$ which is
closed under forming convex mixtures, in the sense that implications of the following form hold
\begin{equation}
  \label{eq:algebraicclosure}
  R(a_1,b_1) \wedge ... \wedge R(a_n,b_n) \Rightarrow R(\sum_i p_i a_i, \sum_i p_i b_i)
\end{equation}
A~\define{state} of an object~$A$ in a monoidal category is a morphism of type~$I \rightarrow A$ where~$I$ is the monoidal unit.
The states in~\convexrel are the convex subsets, as we may have hoped. This model was used as the mathematical basis for a compositional
model of cognition~\cite{BoltCoeckeGenoveseLewisMarsdenPiedeleu2016}.

As our second example, we return to the mathematics of cognition. It is natural to think about notions of nearness and distance for models of reasoning,
a wolf is nearly a dog, a squirrel is closer to being a rat than an elephant. We would therefore like to capture metrics
within our model. We now consider how to introduce metrics into a compact
closed setting. The construction used in the previous example is not applicable as the various natural categories of
metric spaces are not regular. Therefore, we will require a new approach, which will entail a small detour.
We begin by introducing the notion of a quantale.
\begin{definition}[Quantale]
  A \define{quantale} is a join complete partial order~$Q$ with a monoid structure~$(\otimes,k)$ satisfying the following distributivity axioms,
  for all~$a,b \in Q$ and~$A,B \subseteq Q$
  \begin{align*}
    a \otimes \left[ \bigvee B \right] &= \bigvee \{ a \otimes b \mid b \in B \}\\
    \left[ \bigvee A \right] \otimes b &= \bigvee \{ a \otimes b \mid a \in A \}
  \end{align*}
  A quantale is said to be~\define{commutative} if its monoid structure is commutative.
\end{definition}
\begin{example}
  Every locale~\cite{Johnstone1986} is a commutative quantale, and in particular any complete chain
  is a commutative quantale with
  \begin{align*}
    \bigvee A &= \sup A\\
    a_1 \otimes a_2 &= \operatorname{min}(a,b)\\
    k &= \top
  \end{align*}
  \begin{itemize}
  \item The \define{Boolean quantale}~\quantale{B} is given by the chain~$\{0,1\}$ with its usual ordering
  \item The \define{interval quantale}~\quantale{I} is given by the chain~$[0,1] \subseteq \mathbb{R}$ with its usual ordering
  \item The quantale~\quantale{F} is given by the chain~$[0,\infty]$ of extended positive reals with the~\emph{reverse} ordering
  \end{itemize}
  An important example of a commutative quantale that does not correspond to a locale is the \define{Lawvere quantale}~\quantale{C} with
  underlying set the extended positive reals with reverse order and algebraic structure
  \begin{align*}
    \bigvee A &= \inf A\\
    a_1 \otimes a_2 &= a_1 + a_2\\
    k &= 0
  \end{align*}
\end{example}
A binary relation between two sets~$A$ and~$B$ can be described by its characteristic function
\begin{equation*}
  A \times B \rightarrow 2
\end{equation*}
where~$2$ is the two element set of Boolean truth values. We can generalize the notion of binary relation
by allowing the truth values to be taken in a suitable choice of quantale~$Q$, as a function of the form
\begin{equation*}
  A \times B \rightarrow Q
\end{equation*}
We can see this as a potentially infinite matrix of truth values. These binary relations form a category~\relq{Q},
with identities and composition given by suitable generalizations of their matrix theoretic analogues. If
the quantale of truth values is commutative, \relq{Q}~is in fact dagger compact closed. So we have found
another dagger compact closed category, but what has this got to do with metrics? In order to establish the
required connection, we note that we can order relations pointwise in the quantale order, as follows:
\begin{equation*}
  R \subseteq R' \quad\text{ iff }\quad \forall a,b. R(a,b) \leq R'(a,b)
\end{equation*}
This order structure makes~\relq{Q} into a poset-enriched symmetric monoidal category. This means we can consider internal monads,
in the sense of formal category theory~\cite{Street1972}. These identify important ``structured objects'' within our categories as follows.
\begin{itemize}
\item The internal monads of~\relq{\quantale{B}} are endo-relations such that
  \begin{equation*}
    R(a,a) \quad\text{ and }\quad R(a,b) \wedge R(b,c) \Rightarrow R(a,c)
  \end{equation*}
  That is, they are preorders.
\item The internal monads of~\relq{\quantale{I}} are endo-relations where
  \begin{equation*}
    R(a,a) = 1 \quad\text{ and }\quad R(a,b) \wedge R(b,c) \leq R(a,c)
  \end{equation*}
  We can see these as a fuzzy generalization of the notion of a preorder.
\item The key example is the internal monads of~\relq{\quantale{C}}. These are endo-relations satisfying
  \begin{equation*}
    R(a,a) = 0 \quad\text{ and }\quad R(a,b) + R(b,c) \geq R(a,c)
  \end{equation*}
  If we read the relation~$R$ as a distance function, we see that they are~\emph{generalized metric spaces} \cite{Lawvere1973}.
\item The internal monads of~\relq{\quantale{F}} are endo-relations satisfying
  \begin{equation*}
    R(a,a) = 0 \quad\text{ and }\quad \operatorname{max}(R(a,b), R(b,c)) \geq R(a,c)
  \end{equation*}
  Again, if we regard~$R$ as a distance function, these can be seen to be~\emph{generalized ultrametric spaces}.
\end{itemize}
So in particular, \relq{\quantale{C}} gives us a partial order enriched dagger compact closed category in which the internal monads
are generalized metric spaces. Such categories of relations have been proposed as a unifying categorical setting for investigating
various topological notions, see~\cite{ClementinoTholen2003, HofmannSealTholen2014}. Multi-valued relations have
also been investigated for compositional models of natural language~\cite{DostalSadrzadeh2016}.

To recap, we have constructed two compact closed categories using differing techniques that can be found in the literature:
\begin{itemize}
\item By exploiting relations respecting algebraic structure, standard monad and regular category theory provided us with a category where the states are exactly convex subsets.
\item Generalizing the notion of relations in a different direction, we produced a category where the internal monads are generalized metric spaces.
\end{itemize}
So, using rather ad-hoc methods, we have solved two modelling problems using generalizations of binary relations. This prompts several questions:
\begin{itemize}
\item How do these constructions relate to each other? In particular, can we simultaneously work with convexity and metrics in an appropriate setting?
\item Can they be seen as instances of a general construction? 
\item Does the notion of binary relation permit further axes of variation, producing additional examples of compact closed categories?
\item As the parameters of our constructions vary, can the resulting categories be related? Formally, this is a question of functoriality in a suitable sense.
\end{itemize}
These questions provide the starting point for our investigations.
We also observe that the categories we identified in our examples are both in fact instances of Fong and Kissinger's hypergraph categories~\cite{Fong2016}.
These are a particularly well behaved class of dagger compact closed categories, and this will be our technical setting for the remainder of the paper.

We summarize our contributions as follows
\begin{itemize}
\item We provide parameterized constructions of hypergraph categories of generalized relations and spans in theorems~\ref{thm:relvqhypergraph}, \ref{thm:simplespans} and~\ref{thm:spanvqhypergraph}.
\item We introduce analogues of the notion of converse of a generalized relation, and taking the graph of an underlying morphism. Many further aspects are shown to commute with this important structure.
\item In section~\ref{sec:orderstructure} the resulting categories are shown to be appropriately order enriched.
\item In section~\ref{sec:fromspanstorelations} we show that generalized spans can be functorially mapped to generalized relations in a manner respecting all the important structure.
\item In section~\ref{sec:changingtruthvalues} we show that homomorphisms of truth values functorially induce functors between models, preserving all the important structure.
\item In section~\ref{sec:algebraicstructure} we show that our constructions are functorial in the choices of algebraic structure. We also describe how the algebraic and truth value structures interact,
  providing connections with notions resource sensitivity in the sense of linear logic.
\item In section~\ref{sec:changingtopos} we show that our constructions are also functorial in the choice of ambient topos, with the quantale structure being transferred along a logical functor.
\item In theorem~\ref{thm:thebox} we show that the functors induced by changes of parameters commute with each other.
\item Our methods give explicit concrete descriptions of the mathematical objects of interest, suitable for use in applications.
\item We provide many examples illustrating the flexibility of our techniques, particularly to the construction of new and existing models of natural language processing and cognition applications.
\end{itemize}

\subsection*{Related Work}
Categories of relations have been studied in the form of allegories~\cite{FreydScedrov1990}.
This work is somewhat removed from our approach, the heavy use of the modular law does not directly
yield the graphical phenomena of interest. Of more direct relevance is the concept of cartesian
bicategory of~\cite{CarboniWalters1987}. Although graphical notation is not used directly in this
work, these categories can be seen as close relatives of the hypergraph categories resulting from our constructions.
The emphasis in the study of cartesian bicategories was characterization rather than construction of models.

A somewhat syntactic approach to constructing categories with graphical calculi is the use of PROPs~\cite{Maclane1965, Lack2004}.
They have recently been used to construct various categorical models relating to control theory~\cite{BonchiSobocinskiZanasi2015, Zanasi2015, Erbele2016}. These
methods begin with syntax and equations, and freely derive a resulting category. This style is most effective when the application under
consideration has well understood calculational properties. Our approach instead emphasizes the direct construction of models which
can then be investigated for their suitability to a given application.

The beautiful work on decorated cospans and correlations of~\cite{Fong2015, Fong2016}, motivated by the program of network theory initiated in~\cite{Baez2011},
is of most direct relevance to our approach. In a precise sense, the decorated corelation construction is completely generic, every
hypergraph category is produced by that construction. Our emphasis is different, we do not aim for
maximum generality. Instead, our aim is conceptually motivated parameterization. By providing four
clearly motivated features that can be adjusted to application needs, we aim for a practical construction
with which investigators using process theories can construct new models with desirable features.

\section{Mathematical Background}
We will be interested in particular types of symmetric monoidal categories, and
will make use of their corresponding graphical languages~\cite{Selinger2010}.
Technical background on monoidal categories and general categorical notions can
be found in~\cite{MacLane1998}. We will also refer to toposes and their internal
languages in places, standard references are~\cite{MacLaneMoerdijk1992, Johnstone2002a, Johnstone2002b, Borceux1994c}.
The paper has been written with the intention that it should be readable without any detailed
knowledge of topos theory. For such readers, definitions should be read as if they pertain to ordinary sets, functions and predicate logic.
In this section we briefly describe some standard mathematical background and conventions.
\begin{definition}[Compact Closed Category]
  An object~$A$ in a symmetric monoidal category is said to have dual~$A^*$ if there
  exist unit~$\eta: I \rightarrow A^* \otimes A$ and counit~$\epsilon: A \otimes A^* \rightarrow I$ morphisms.
  These morphisms are depicted in the graphical calculus using special notation as
  
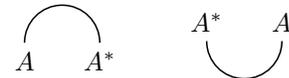
\begin{equation*}
    \begin{gathered}
      \begin{tikzpicture}[scale=0.5, stringdiagram]
        \path coordinate (cap)
        +(-1,-1) coordinate[label=below:$A$] (bl)
        +(1,-1) coordinate[label=below:$A^*$] (br);
        \draw (bl) to[out=90, in=180] (cap.west) -- (cap.east) to[out=0, in=90] (br);
      \end{tikzpicture}
    \end{gathered}
    \qquad
    \begin{gathered}
      \begin{tikzpicture}[scale=0.5, stringdiagram]
        \path coordinate (cup)
        +(-1,1) coordinate[label=above:$A^*$] (tl)
        +(1,1) coordinate[label=above:$A$] (tr);
        \draw (tl) to[out=-90, in=180] (cup.west) -- (cup.east) to[out=0, in=-90] (tr);
      \end{tikzpicture}
    \end{gathered}
  \end{equation*}
  They are required to satisfy the following~\define{snake equations}.
  
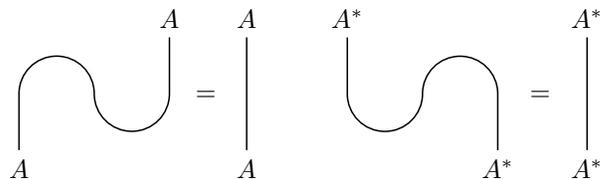
\begin{equation*}
    \begin{gathered}
      \begin{tikzpicture}[scale=0.5, stringdiagram]
        \path coordinate[label=below:$A$] (bot)
        ++(0,1.5) coordinate (a)
        ++(1,1) coordinate (cap)
        ++(1,-1) coordinate (b)
        ++(1,-1) coordinate (cup)
        ++(1,1) coordinate (c)
        ++(0,1.5) coordinate[label=above:$A$] (top);
        \draw (bot) -- (a) to[out=90, in=180] (cap.west) -- (cap.east)
        to[out=0, in=90] (b.north) -- (b.south) to[out=-90, in=180] (cup.west) -- (cup.east)
        to[out=0, in=-90] (c.south) -- (c.north) -- (top);
      \end{tikzpicture}
    \end{gathered}
    =
    \begin{gathered}
      \begin{tikzpicture}[scale=0.5, stringdiagram]
        \path coordinate[label=below:$A$] (bot) ++(0,3) coordinate[label=above:$A$] (top);
        \draw (bot) -- (top);
      \end{tikzpicture}
    \end{gathered}
    \qquad
    \begin{gathered}
      \begin{tikzpicture}[scale=0.5, stringdiagram]
        \path coordinate[label=above:$A^*$] (top)
        ++(0,-1.5) coordinate (a)
        ++(1,-1) coordinate (cup)
        ++(1,1) coordinate (b)
        ++(1,1) coordinate (cap)
        ++(1,-1) coordinate (c)
        ++(0,-1.5) coordinate[label=below:$A^*$] (bot);
        \draw (top) -- (a) to[out=-90, in=180] (cup.west) -- (cup.east)
        to[out=0, in=-90] (b.south) -- (b.north) to[out=90, in=180] (cap.west) -- (cap.east)
        to[out=0, in=90] (c.north) -- (c.south) -- (bot);
      \end{tikzpicture}
    \end{gathered}
    =
    \begin{gathered}
      \begin{tikzpicture}[scale=0.5, stringdiagram]
        \path coordinate[label=below:$A^*$] (bot) ++(0,3) coordinate[label=above:$A^*$] (top);
        \draw (bot) -- (top);
      \end{tikzpicture}
    \end{gathered}
  \end{equation*}
  A~\define{compact closed category} is a symmetric monoidal category in which every object has
  a dual. A compact closed category~$\mathcal{A}$, equipped with an identity on objects
  involution~$(-)^\dagger : \mathcal{A}^{op} \rightarrow \mathcal{A}$ coherently with the symmetric
  monoidal compact closed structure, is referred to as a \define{dagger compact closed category}\cite{AbramskyCoecke2004}.
  The older term~\define{strongly compact closed category} is also occasionally used.
\end{definition}
\begin{example}
  \label{ex:reldaggercompactclosed}
  The canonical example of a dagger compact closed category of relevance to the current work is the category~\rel of
  sets and binary relations between them. The symmetric monoidal structure is given by cartesian products of sets, and
  the dagger by the usual converse of relations.
  Objects are self-dual, with the unit on a set~$A$ given by the relation
  \begin{equation*}
    \eta = \{ (*, (a,a)) \mid a \in A \}
  \end{equation*}
  and the counit is its converse.
\end{example}
\begin{definition}[Hypergraph Category]
  A~\define{hypergraph category} is a symmetric monoidal category such that every object~$A$ carries
  a commutative monoid structure
  \begin{equation*}
    (\eta : I \rightarrow A, \mu: A \otimes A \rightarrow A)
  \end{equation*}
  and a cocommutative comonoid structure
  \begin{equation*}
    (\epsilon: A \rightarrow I, \delta : A \rightarrow A \otimes A)
  \end{equation*}
  We depict these morphisms graphically as follows:
  
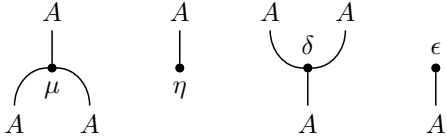
\begin{equation*}
    \begin{tikzpicture}[scale=0.5, stringdiagram, baseline=(mu.center)]
      \path coordinate[dot, label=below:$\mu$] (mu)
      +(-1,-1) coordinate[label=below:$A$] (bl)
      +(1,-1) coordinate[label=below:$A$] (br)
      +(0,1) coordinate[label=above:$A$] (top);
      \draw
      (mu) -- (top)
      (bl) to[out=90, in=180] (mu.west) -- (mu.east) to[out=0, in=90] (br);
    \end{tikzpicture}
    \qquad
    \begin{tikzpicture}[scale=0.5, stringdiagram, baseline=(eta.center)]
      \path coordinate[dot, label=below:$\eta$] (eta)
      +(0,1) coordinate[label=above:$A$] (top);
      \draw (eta) -- (top);
    \end{tikzpicture}
    \qquad
    \begin{tikzpicture}[scale=0.5, stringdiagram, baseline=(delta.center)]
      \path coordinate[dot, label=above:$\delta$] (delta)
      +(-1,1) coordinate[label=above:$A$] (tl)
      +(1,1) coordinate[label=above:$A$] (tr)
      +(0,-1) coordinate[label=below:$A$] (bot);
      \draw
      (delta) -- (bot)
      (tl) to[out=-90, in=180] (delta.west) -- (delta.east) to[out=0, in=-90] (tr);
    \end{tikzpicture}
    \qquad
    \begin{tikzpicture}[scale=0.5, stringdiagram, baseline=(epsilon.center)]
      \path coordinate[dot, label=above:$\epsilon$] (epsilon)
      +(0,-1) coordinate[label=below:$A$] (bot);
      \draw (epsilon) -- (bot);
    \end{tikzpicture}
  \end{equation*}
  The choice of monoid structure on each object is required to satisfy the following coherence
  condition with respect to the monoidal structure.
  
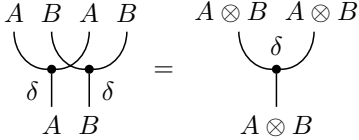
\begin{equation*}
    \label{eq:hypergraphcoherence}
    \begin{gathered}
      \begin{tikzpicture}[xscale=0.5, yscale=-0.5, stringdiagram]
        \path coordinate[dot, label=below left:$\delta$] (mul)
        +(0,1) coordinate[label=below:$A$] (tl)
        +(-1,-1) coordinate[label=above:$A$] (a)
        +(1,-1) coordinate[label=above:$A$] (b)
        ++(1,0) coordinate[dot, label=below right:$\delta$] (mur)
        +(0,1) coordinate[label=below:$B$] (tr)
        +(-1,-1) coordinate[label=above:$B$] (c)
        +(1,-1) coordinate[label=above:$B$] (d);
        \draw
        (a) to[out=90, in=180] (mul.west) -- (mul.east) to[out=0, in=90] (b)
        (d) to[out=90, in=0] (mur.east) -- (mur.west) to[out=180, in=90] (c)
        (mul) -- (tl)
        (mur) -- (tr);
      \end{tikzpicture}
    \end{gathered}
    =
    \begin{gathered}
      \begin{tikzpicture}[xscale=0.5, yscale=-0.5, stringdiagram]
        \path coordinate[dot, label=above:$\delta$] (mu)
        +(0,1) coordinate[label=below:$A \otimes B$] (top)
        +(-1,-1) coordinate[label={[xshift=-1mm]above:$A \otimes B$}] (a)
        +(1,-1) coordinate[label={[xshift=1mm]above:$A \otimes B$}] (b);
        \draw (a) to[out=90, in=180] (mu.west) -- (mu.east) to[out=0, in=90] (b)
        (mu) -- (top);
      \end{tikzpicture}
    \end{gathered}
  \end{equation*}
  Here, we overload the use of the symbol~$\delta$ to avoid cluttering our diagrams with indices
  or subscripts. We will exploit similar overloading of names in many places in what follows.
  The monoid structure is also required to satisfy the dual coherence condition.
  The multiplication and comultiplication must also satisfy the Frobenius axiom
  
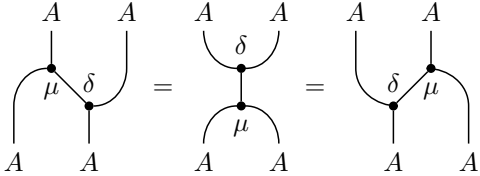
\begin{equation*}
    \begin{gathered}
      \begin{tikzpicture}[scale=0.5, stringdiagram]
        \path coordinate[dot, label=below:$\mu$] (mu)
        +(-1,-1) coordinate (a)
        +(0,1) coordinate[label=above:$A$] (tl)
        ++(1,-1) coordinate[dot, label=above:$\delta$] (delta)
        +(1,1) coordinate (b)
        +(0,-1) coordinate[label=below:$A$] (br)
        (a) ++(0,-1) coordinate[label=below:$A$] (bl)
        (b) ++(0,1) coordinate[label=above:$A$] (tr);
        \draw
        (mu) -- (delta)
        (bl) -- (a.south) -- (a.north) to[out=90, in=180] (mu.west)
        (tr) -- (b.north) -- (b.south) to[out=-90, in=0] (delta.east)
        (mu) -- (tl)
        (delta) -- (br);
      \end{tikzpicture}
    \end{gathered}
    =
    \begin{gathered}
      \begin{tikzpicture}[scale=0.5, stringdiagram]
        \path coordinate[dot, label=below:$\mu$] (mu)
        +(-1,-1) coordinate[label=below:$A$] (bl)
        +(1,-1) coordinate[label=below:$A$] (br)
        ++(0,1) coordinate[dot, label=above:$\delta$] (delta)
        +(-1,1) coordinate[label=above:$A$] (tl)
        +(1,1) coordinate[label=above:$A$] (tr);
        \draw
        (tl) to[out=-90, in=180] (delta.west) -- (delta.east) to[out=0, in=-90] (tr)
        (bl) to[out=90, in=180] (mu.west) -- (mu.east) to[out=0, in=90] (br)
        (mu) -- (delta);            
      \end{tikzpicture}
    \end{gathered}
    =
    \begin{gathered}
      \begin{tikzpicture}[xscale=-0.5, yscale=0.5, stringdiagram]
        \path coordinate[dot, label=below:$\mu$] (mu)
        +(-1,-1) coordinate (a)
        +(0,1) coordinate[label=above:$A$] (tl)
        ++(1,-1) coordinate[dot, label=above:$\delta$] (delta)
        +(1,1) coordinate (b)
        +(0,-1) coordinate[label=below:$A$] (br)
        (a) ++(0,-1) coordinate[label=below:$A$] (bl)
        (b) ++(0,1) coordinate[label=above:$A$] (tr);
        \draw
        (mu) -- (delta)
        (bl) -- (a.south) -- (a.north) to[out=90, in=180] (mu.west)
        (tr) -- (b.north) -- (b.south) to[out=-90, in=0] (delta.east)
        (mu) -- (tl)
        (delta) -- (br);
      \end{tikzpicture}
    \end{gathered}
  \end{equation*}
  and the special axiom
  
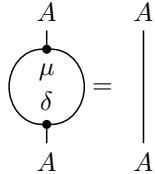
\begin{equation*}
    \begin{gathered}
      \begin{tikzpicture}[scale=0.5, stringdiagram]
        \path coordinate[dot, label=above:$\delta$] (delta)
        +(0,-0.5) coordinate[label=below:$A$] (bot)
        +(-1,1) coordinate (a)
        ++(1,1) coordinate (b)
        ++(-1,1) coordinate[dot, label=below:$\mu$] (mu)
        +(0,0.5) coordinate[label=above:$A$] (top);
        \draw
        (bot) -- (delta)
        (top) -- (mu)
        (a.south) to[out=-90, in=180] (delta.west) -- (delta.east) to[out=0, in=-90] (b.south) --
        (b.north) to[out=90, in=0] (mu.east) -- (mu.west) to[out=180, in=90] (a.north);
      \end{tikzpicture}
    \end{gathered}
    =
    \begin{gathered}
      \begin{tikzpicture}[scale=0.5, stringdiagram]
        \path coordinate[label=below:$A$] (bot) ++(0,3) coordinate[label=above:$A$] (top);
        \draw (bot) -- (top);
      \end{tikzpicture}
    \end{gathered}
  \end{equation*}
  More briefly, a hypergraph category is a symmetric monoidal category with
  a chosen special commutative Frobenius algebra structure on every object, coherent with the tensor product.
\end{definition}
\begin{example}
  The category~\rel is also an example of a hypergraph category. The cocommutative comonoid is given by the relations
  \begin{equation*}
    \epsilon = \{ (a,*) \mid a \in A \} \qquad \delta = \{ (a,(a,a)) \mid a \in A \}
  \end{equation*}
  The monoid is the relational converse of the comonoid structure. The induced dagger compact closed structure
  is exactly that of example~\ref{ex:reldaggercompactclosed}.
\end{example}
Our interest in hypergraph categories is that they are a particularly pleasant form of dagger compact closed category,
as established by the following well known observation.
\begin{proposition}
  Every hypergraph category is a dagger compact closed category, with the cup and cap given by
  
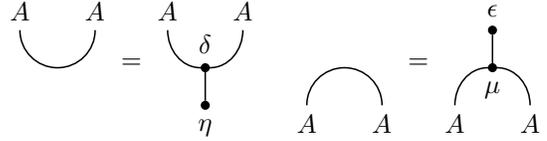
\begin{equation*}
    \begin{tikzpicture}[scale=0.5, stringdiagram, baseline=(eta.center)]
      \path coordinate (eta)
      +(-1,1) coordinate[label=above:$A$] (tl)
      +(1,1) coordinate[label=above:$A$] (tr);
      \draw (tl) to[out=-90, in=180] (eta.west) -- (eta.east) to[out=0, in=-90] (tr);
    \end{tikzpicture}
    =
    \begin{tikzpicture}[scale=0.5, stringdiagram, baseline=(delta.center)]
      \path coordinate[dot, label=above:$\delta$] (delta)
      +(0,-1) coordinate[dot, label=below:$\eta$] (eta)
      +(-1,1) coordinate[label=above:$A$] (tl)
      +(1,1) coordinate[label=above:$A$] (tr);
      \draw
      (tl) to[out=-90, in=180] (delta.west) -- (delta.east) to[out=0, in=-90] (tr)
      (delta) -- (eta);
    \end{tikzpicture}
    \quad
    \begin{tikzpicture}[scale=0.5, stringdiagram, baseline=(epsilon.center)]
      \path coordinate (epsilon)
      +(-1,-1) coordinate[label=below:$A$] (bl)
      +(1,-1) coordinate[label=below:$A$] (br);
      \draw (bl) to[out=90, in=180] (epsilon.west) -- (epsilon.east) to[out=0, in=90] (br);
    \end{tikzpicture}
    =
    \begin{tikzpicture}[scale=0.5, stringdiagram, baseline=(mu.center)]
      \path coordinate[dot, label=below:$\mu$] (mu)
      +(0,1) coordinate[dot, label=above:$\epsilon$] (epsilon)
      +(-1,-1) coordinate[label=below:$A$] (bl)
      +(1,-1) coordinate[label=below:$A$] (br);
      \draw
      (bl) to[out=90, in=180] (mu.west) -- (mu.east) to[out=0, in=90] (br)
      (mu) -- (epsilon);
    \end{tikzpicture}
  \end{equation*}
  The dagger of a morphism~$f : A \rightarrow B$ is given by its \define{transpose}
  
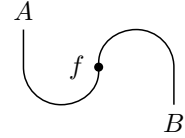
\begin{equation*}
    \begin{tikzpicture}[scale=0.5, stringdiagram]
      \path coordinate[dot, label=left:$f$] (f)
      +(1,1) coordinate (epsilon)
      +(-1,-1) coordinate (eta)
      (epsilon) ++(1,-1) coordinate (a) ++(0,-1) coordinate[label=below:$B$] (bot)
      (eta) ++(-1,1) coordinate (b) ++(0,1) coordinate[label=above:$A$] (top);
      \draw (bot) -- (a.south) -- (a.north) to[out=90, in=0] (epsilon.east) -- (epsilon.west) to[out=180, in=90] (f.north) -- (f.south)
      to[out=-90, in=0] (eta.east) -- (eta.west) to[out=180, in=-90] (b.south) -- (b.north) -- (top);
    \end{tikzpicture}
  \end{equation*}
\end{proposition}
As a final technical point, we will be working with various categories with finite products. Throughout, we will implicitly
assume a choice of terminal object and binary products has been given. To reduce clutter, we therefore resist repeating this assumption in
the statements of our subsequent theorems.

\section{Relations}
\label{sec:algebraicqrelations}
The aim in this section is to broadly generalize the notion of binary relation between sets, in order to support our motivating examples, and to provide
scope for many other variations. We observed, for sets~$A$ and~$B$, and
quantale~$Q$, that we can consider a function~$A \times B \rightarrow Q$ as a relation, with truth values taken in
the quantale. For such generalized relations, we define the composition of relations~$R : A \rightarrow B$ and~$S : B \rightarrow C$
by analogy with the usual composition of relations
\begin{equation*}
  (S \circ R)(a,c) = \bigvee_b R(a,b) \otimes S(b,c)
\end{equation*}
With this notion of composition, the following relation, with truth values in~$Q$, serves as an identity on set~$A$:
\begin{equation*}
  1_A(a_1,a_2) = \bigvee \{ k \mid a_1 = a_2 \}
\end{equation*}
We then observe that all of these definitions actually make sense in the internal language of an arbitrary elementary topos. This
leads us to the following definition.
\begin{definition}[$Q$-relation]
  \label{def:qrelation}
  Let~$\mathcal{E}$ be a topos, and~$(Q,\otimes,k,\bigvee)$ an internal quantale. A~\define{$Q$-relation} between~$\mathcal{E}$
  objects~$A$ and~$B$ is a~$\mathcal{E}$-morphism of type
  \begin{equation*}
    A \times B \rightarrow Q
  \end{equation*}
  $\mathcal{E}$-objects and $Q$-relations between them form a category~\relq{Q}, with identities and composition as described above.
\end{definition}
Definition~\ref{def:qrelation} is a first step in the right direction, but in order to capture convexity, as discussed in the introduction,
we must find a way of incorporating algebraic structure. If we consider an algebraic signature~$(\Sigma,E)$
with set of operations~$\Sigma$ and equations~$E$, the general form of equation~\eqref{eq:algebraicclosure},
for $n$-ary operation~$\sigma \in \Sigma$, is
\begin{equation*}
  R(a_1,b_1) \wedge ... \wedge R(a_n, b_n) \Rightarrow R(\sigma(a_1,...,a_n),\sigma(b_1,...,b_n))
\end{equation*}
We will require throughout that all operation symbols have finite arity, as is conventional in universal algebra.

It is then natural to consider replacing the logical components of this definition with the structure of our chosen
quantale. This leads to the definition we require.
\begin{definition}[Algebraic $Q$-relation]
  \label{def:algebraicqrelation}
  Let~$\mathcal{E}$ be a topos, and~$(Q, \otimes, k, \bigvee)$ an internal quantale.
  Let~$(\Sigma,E)$ be an algebraic variety in~$\mathcal{E}$.
  An~\define{algebraic}~$Q$-relation between~$(\Sigma,E)$-algebras~$A$ and~$B$ is a~$Q$-relation between
  their underlying~$\mathcal{E}$-objects such that for each~$\sigma \in \Sigma$ the following axiom holds
  \begin{align}
    \label{eq:crucialinequation}
    R(a_1,b_1) \otimes ... \otimes &R(a_n,b_n) \nonumber \\
    &\leq R(\sigma(a_1,...,a_n), \sigma(b_1,...,b_n))
  \end{align}
  $(\Sigma,E)$-algebras and algebraic~$Q$-relations form a category~\relvq{(\Sigma,E)}{Q}, with identities and composition as for their
  underlying~$Q$-relations.
\end{definition}
There is some subtlety to the interaction between truth values and algebraic structure, we will return to this
topic in section~\ref{sec:algebraicstructure}. We now continue studying the categorical structure of algebraic~$Q$-relations.
\begin{restatable}{proposition}{relisacategory}\label{prop:relisacategory}
  Let~$\mathcal{E}$ be a topos, $(\Sigma,E)$ a variety in~$\mathcal{E}$, and~$(Q, \otimes, k, \bigvee)$ an internal commutative quantale.
  The category~\relvq{(\Sigma,E)}{Q} is a symmetric monoidal category. The symmetric monoidal structure is inherited from
  the finite products in~$\mathcal{E}$.
\end{restatable}

We can take the converse of an ordinary binary relation, simply by reversing its arguments. The notion of converse
generalizes smoothly to algebraic~$Q$-relations, in a manner that respects all the relevant categorical structure.
\begin{restatable}{proposition}{relconverse}[Converse]\label{prop:relconverse}
  \label{prop:relationconverse}
  Let~$\mathcal{E}$ be a topos, $(\Sigma,E)$ a variety in~$\mathcal{E}$, and~$(Q, \otimes, k, \bigvee)$ an internal commutative quantale.
  There is an identity on objects strict symmetric monoidal functor
  \begin{equation*}
    \converse{(-)} : \relvq{(\Sigma,E)}{Q}^{op} \rightarrow \relvq{(\Sigma,E)}{Q}
  \end{equation*}
  given by reversing arguments:
  \begin{equation*}
    \converse{R}(b,a) = R(a,b)
  \end{equation*}
\end{restatable}
For ordinary sets and functions,
given a function
\begin{equation*}
  f :A \rightarrow B
\end{equation*}
we can form a binary relation using the graph of~$f$
\begin{equation*}
  \{ (a,b) \mid f(a) = b \}
\end{equation*}
The next proposition establishes that we can take graphs of morphisms in our underlying category of algebras,
in a manner respecting all the relevant categorical structure.
\begin{restatable}{proposition}{relgraph}[Graph]\label{prop:relgraph}
  Let~$\mathcal{E}$ be a topos, $(\Sigma,E)$ a variety in~$\mathcal{E}$, and~$(Q, \otimes, k, \bigvee)$ an internal commutative quantale.
  There is an identity on objects strict symmetric monoidal functor
  \begin{equation*}
    \graph{(-)} : \sigmaalg{\Sigma}{E} \rightarrow \relvq{(\Sigma,E)}{Q}
  \end{equation*}
  defined on morphism~$f : A \rightarrow B$ by
  \begin{equation*}
    f_\circ(a,b) = \bigvee \{ k \mid f(a) = b \}
  \end{equation*}
  The symmetric monoidal structure on~\sigmaalg{\Sigma}{E} is the finite product structure.
\end{restatable}
The graph functor allows us to lift structures from the underlying category of algebras. The following
canonical comonoids are of particular conceptual importance.
\begin{restatable}{proposition}{canonicalcomonoid}
  \label{prop:canonicalcomonoid}
  Let~$\mathcal{E}$ be a category with finite products. Each object~$A$ carries a cocommutative comonoid structure via the canonical morphisms
  \begin{equation*}
    ! : A \rightarrow 1 \quad\text{ and }\quad \langle 1_A, 1_A \rangle : A \rightarrow A \times A
  \end{equation*}
  These morphisms satisfy the coherence condition~\eqref{eq:hypergraphcoherence}.
\end{restatable}
Finally, we are in a position to establish that our categories of algebraic~$Q$-relations are hypergraph categories.
\begin{restatable}{theorem}{relvqhypergraph}
  \label{thm:relvqhypergraph}
  Let~$\mathcal{E}$ be a topos, $(\Sigma,E)$ a variety in~$\mathcal{E}$, and~$(Q, \otimes, k, \bigvee)$ an internal commutative quantale.
  The category~\relvq{(\Sigma,E)}{Q} is a hypergraph category. The cocommutative comonoid structure is given by the graphs of
  the canonical comonoids described in proposition~\ref{prop:canonicalcomonoid}, and the monoid structure is given by
  their converses.
\end{restatable}
We quickly return to one of the examples discussed in the introduction.
\begin{example}
  \label{ex:convex}
  The convex algebras discussed in the introduction can be presented by a family of binary operations for forming pairwise convex combinations
  \begin{equation*}
    +^p \text{ where } p \in (0,1)
  \end{equation*}
  satisfying suitable equations. Writing~\signature{Convex} for this signature,
  we can construct~\convexrel as~\relvq{\signature{Convex}}{2}, where~$2$ is the two element set.
\end{example}

\section{Spans}
\label{sec:spans}
Generalizing the truth values, algebraic structure and ambient category has provided three degrees of freedom
for describing custom hypergraph categories. Currently we can vary the underlying topos, quantale and choice of algebraic structure.
We now investigate a fourth, final direction of variation.

If we consider a span of sets
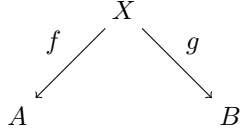
\begin{equation*}
  \begin{tikzpicture}[scale=0.5, node distance=2cm, ->]
    \node (apex) {$X$};
    \node[below left of=apex] (bl) {$A$};
    \node[below right of=apex] (br) {$B$};
    \draw (apex) to node[above left]{$f$} (bl);
    \draw (apex) to node[above right]{$g$} (br);
  \end{tikzpicture}
\end{equation*}
we can consider an element~$x \in X$ as a proof witness relating~$f(x)$ and~$g(x)$.
Two spans are composed by forming the pullback
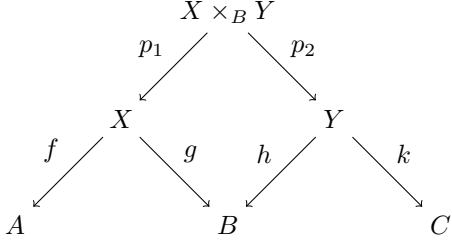
\begin{equation*}
  \begin{tikzpicture}[scale=0.5, node distance=2cm, ->]
    \node (bl) {$A$};
    \node[above right of=bl] (x) {$X$};
    \node[above right of=x] (p) {$X \times_B Y$};
    \node[below right of=x] (bm) {$B$};
    \node[above right of=bm] (y) {$Y$};
    \node[below right of=y] (br) {$C$};
    \draw (x) to node[above left]{$f$} (bl);
    \draw (x) to node[above right]{$g$} (bm);
    \draw (y) to node[above left]{$h$} (bm);
    \draw (y) to node[above right]{$k$} (br);
    \draw (p) to node[above left]{$p_1$} (x);
    \draw (p) to node[above right]{$p_2$} (y);
  \end{tikzpicture}
\end{equation*}
Recall that pullbacks in~\cset are given explicitly by
\begin{equation*}
  X \times_B Y = \{ (x,y) \mid g(x) = h(y) \}
\end{equation*}
Therefore, a pair~$(x,y)$ relates $a$ and~$c$ exactly if~$x$ relates~$a$ to some~$b$ and this~$b$ is related to~$c$ by~$y$. So, at least for the category~\cset,
we can think of spans as proof relevant relations. This is the intuition we now pursue, starting by adjusting the
notion of~$Q$-relation in definition~\ref{def:qrelation} to the setting of spans.
\begin{definition}[$Q$-span]
  \label{def:qspan}
  Let~$\mathcal{E}$ be a finitely complete category, and~$(Q, \otimes, k)$ an internal monoid. A \define{$Q$-span}
  of type~$A \rightarrow B$ is a quadruple~$(X,f,g,\chi)$ where
  \begin{itemize}
  \item $(X,f : X \rightarrow A, g : X \rightarrow B)$ is a span in~$\mathcal{E}$
  \item $\chi : X \rightarrow Q$ is a~$\mathcal{E}$-morphism, referred to as the~\define{characteristic morphism}.
  \end{itemize}
  Two $Q$-spans $(X,f, g, \chi), (Y,h,k,\xi)$ are composed by composing their underlying spans by pullback, and taking the resulting characteristic
  morphism to be
  \begin{equation*}
    X \times_C Y \xrightarrow{ \langle p_1, p_2 \rangle} X \times Y \xrightarrow{\chi \times \xi} Q \times Q \xrightarrow{ \otimes } Q 
  \end{equation*}
  where~$p_1$ and~$p_2$ are the pullback projections.
  
  A~\define{morphism of~$Q$-spans} between two~$Q$-spans of type~$A \rightarrow B$
  \begin{equation*}
    \alpha : (X_1, f_1, g_1, \chi_1) \rightarrow (X_2, f_2, g_2, \chi_2)
  \end{equation*}
  is a~$\mathcal{E}$-morphism~$\alpha : X_1 \rightarrow X_2$ such that
  \begin{align*}
    f_1 = f_2 \circ \alpha \qquad
    g_1 = g_2 \circ \alpha \qquad
    \chi_1 = \chi_2 \circ \alpha
  \end{align*}
\end{definition}
\begin{remark}
  When discussing $Q$-spans in the remainder of this paper, we actually intend isomorphism classes of spans with respect to
  the homomorphisms of definition~\ref{def:qspan}. This convention is common when considering categories of ordinary spans,
  where composition of spans via pullback is only defined up to isomorphism.
  All definitions and calculations using representatives will respect this isomorphism structure. These
  isomorphism classes of~$Q$-spans form a category~\spanq{Q}. If we write~$\chi_k$ for the constant morphism
  \begin{equation*}
    \chi_k = A \xrightarrow{!} 1 \xrightarrow{k} Q
  \end{equation*}
  then the identity at~$A$ is given by the~$Q$-span~$(A,1,1,\chi_k)$.
\end{remark}
These spans with configurable truth values provide another construction of hypergraph categories.
\begin{restatable}{theorem}{simplespans}
  \label{thm:simplespans}
  Let~$\mathcal{E}$ be a finitely complete category, and~$(Q, \otimes, k)$ an internal commutative monoid.
  The category~\spanq{Q} is a hypergraph category.
\end{restatable}
We will not dwell on the explicit symmetric monoidal and hypergraph structures claimed in theorem~\ref{thm:simplespans}.
Once we have incorporated algebraic structure into our span constructions,
the required details can be found in proposition~\ref{prop:spanvqsmc} and theorem~\ref{thm:spanvqhypergraph}.

The key step now is incorporate algebraic structure into the picture, paralleling the ideas of definition~\ref{def:algebraicqrelation}.
In this case, things are slightly more complicated as we have to explicitly administer the proof witnesses in the spans. We also
must introduce an ordering on our truth values in order to specify the necessary axiom.
\begin{definition}
  Let~$\mathcal{E}$ be a topos, $(\Sigma,E)$ a variety in~$\mathcal{E}$, and~$(Q, \otimes, k, \leq)$ an internal partially ordered commutative monoid.
  For~$(\Sigma,E)$-algebras~$A$ and~$B$, an~\define{algebraic}~$Q$-span is a quadruple~$(X,f,g,\chi)$ which is a~$Q$-span between
  the underlying~$\mathcal{E}$-objects satisfying the following axiom.

  For every~$\sigma \in \Sigma$ if
  \begin{equation*}
    f(x_1) = a_1 \wedge g(x_1) = b_1 \wedge ... \wedge f(x_n) = a_n \wedge g(x_n) = b_n
  \end{equation*}
  then there exists~$x$ such that
  \begin{equation*}
    f(x) = \sigma(a_1,...,a_n) \wedge g(x) = \sigma(b_1,...,b_n)
  \end{equation*}
  and
  \begin{equation*}
    \chi(x_1) \otimes ... \otimes \chi(x_n) \leq \chi(x)
  \end{equation*}
  $(\Sigma,E)$-algebras and algebraic~$Q$-spans form a category~\spanvq{(\Sigma,E)}{Q} with identities and composition
  given as for the underlying~$Q$-spans.
\end{definition}
As with the algebraic~$Q$-relations in section~\ref{sec:algebraicqrelations}, we obtain
a symmetric monoidal category with analogues of relational converse and taking graphs.
\begin{restatable}{proposition}{spanvqsmc}
  \label{prop:spanvqsmc}
  Let~$\mathcal{E}$ be a topos, $(\Sigma,E)$ a variety in~$\mathcal{E}$, and~$(Q, \otimes, k, \leq)$ an internal partially ordered commutative monoid.
  The category~\spanvq{(\Sigma,E)}{Q} is a symmetric monoidal category. The symmetric monoidal structure is inherited from the
  finite product structure in~$\mathcal{E}$.
\end{restatable}
\begin{restatable}{proposition}{conversealgebraic}[Converse]
  \label{prop:spanconverse}
  Let~$\mathcal{E}$ be a topos, $(\Sigma,E)$ a variety in~$\mathcal{E}$, and~$(Q, \otimes, k, \leq)$ an internal partially ordered commutative monoid.
  There is an identity on objects strict symmetric monoidal functor
  \begin{equation*}
    \converse{(-)} : \spanvq{(\Sigma,E)}{Q}^{op} \rightarrow \spanvq{(\Sigma,E)}{Q}
  \end{equation*}
  given by reversing the legs of the underlying span:
  \begin{equation*}
    \converse{(X,f,g,\chi)} = (X,g,f,\chi)
  \end{equation*}
\end{restatable}
\begin{restatable}{proposition}{graphalgebraic}[Graph]
  Let~$\mathcal{E}$ be a topos, and~$(Q, \otimes, k, \leq)$ an internal partially ordered commutative monoid.
  There is an identity on objects strict symmetric monoidal functor
  \begin{equation*}
    \graph{(-)} : \sigmaalg{\Sigma}{E} \rightarrow \spanvq{(\Sigma,E)}{Q}
  \end{equation*}
  with the action on morphism~$f : A \rightarrow B$ given by
  \begin{equation*}
    \graph{f} = (A,1,f,\chi_k)
  \end{equation*}
\end{restatable}
As before, we can exploit the graph construction and the canonical comonoids of proposition~\ref{prop:canonicalcomonoid} to
establish the existence of a hypergraph structure.
\begin{restatable}{theorem}{spanvqhypergraph}
  \label{thm:spanvqhypergraph}
  Let~$\mathcal{E}$ be a topos, $(\Sigma,E)$ a variety in~$\mathcal{E}$, and~$(Q, \otimes, k, \leq)$ an internal partially ordered commutative monoid.
  The category~\spanvq{(\Sigma,E)}{Q} is a hypergraph category. The cocommutative comonoid structure is given by the graphs of
  the canonical comonoids described in proposition~\ref{prop:canonicalcomonoid}, and the monoid structure is given by
  their converses.
\end{restatable}
This construction presents new modelling possibilities, that can be combined with other features, opening fresh
directions for investigation that may not have been immediately apparent.
\begin{example}
  \label{ex:convexspan}
  The span construction allows us to construct variations on the models we are already interested in. For example, we can
  now construct a proof relevant version of the model~\ref{ex:convex}. From a practical perspective, this presents
  the possibility of models in which we can describe the interaction of cognitive phenomena, and provide quantitative
  evidence for any relationships that we conclude hold.
\end{example}

\begin{example}
  \label{ex:presheaf}
  Instead of using~\cset as our base topos in our models, we could consider using a presheaf topos~$[\mathcal{C}^{op}, \cset]$
  for a small category~$\mathcal{C}$. This allows us to construct models using ``sets varying with context'', incorporating all the
  features discussed in the previous examples. In linguistic or cognitive examples, contexts could describe time, the agents
  involved or the broader setting in which meaning should be interpreted. These context sensitive models present a lot of new expressive potential,
  and will be investigated in future work.
\end{example}

\section{Order Enrichment}
\label{sec:orderstructure}
In order to meaningfully discuss internal monads, we require some 2-categorical structure on our relational constructions.
Specifically, we introduce an appropriate ordering on our morphisms.
Order enrichment is also important from a practical perspective when modelling real world applications.
For example, in natural language applications, we are often interested in phenomena such as
ambiguity~\cite{Piedeleu2014, PiedeleuKartsaklisCoeckeSadrzadeh2015} and lexical entailment~\cite{Bankova2015},
and these are best studied from an order theoretic perspective.

Generalizing the situation for ordinary set theoretic binary relations, we introduce an ordering on $Q$-relations.
\begin{definition}
  Let~$\mathcal{E}$ be a topos and~$(Q, \otimes, k, \bigvee)$ an internal quantale. We define a partial order on~$Q$-relations as follows
  \begin{equation*}
    R \subseteq R' \quad\text{ iff }\quad \forall a,b. R(a,b) \leq R'(a,b)
  \end{equation*}
  Algebraic~$Q$-relations are ordered similarly, by comparing their underlying~$Q$-relations.
\end{definition}
\begin{restatable}{theorem}{relcanbeordered}
  Let~$\mathcal{E}$ be a topos, $(\Sigma,E)$ a variety in~$\mathcal{E}$, and~$(Q, \otimes, k, \bigvee)$ an internal commutative quantale.
  The category~\relvq{(\Sigma,E)}{Q} is a partially ordered symmetric monoidal category.
\end{restatable}
$Q$-spans can also be ordered, in a manner analogous to that for relations, but explicitly taking into account
the proof witnesses.
\begin{definition}
  For topos~$\mathcal{E}$ and internal partially ordered monoid~$Q$, we define a preorder on~$Q$-spans as follow.
  \begin{equation*}
    (X_1,f_1,g_1,\chi_1) \subseteq (X_2,f_2,g_2,\chi_2)
  \end{equation*}
  if there is a~$\mathcal{E}$-monomorphism~$m : X_1 \rightarrow X_2$ such that
  \begin{equation*}
    f_1 = f_2 \circ m \text{ and } g_1 = g_2 \circ m \text{ and } \forall x. \chi_1(x) \leq \chi_2(m(x))
  \end{equation*}
  Algebraic~$Q$-spans are ordered similarly, by comparing their underlying~$Q$-spans.
\end{definition}
\begin{restatable}{theorem}{spancanbeordered}
  Let~$\mathcal{E}$ be a topos, $(\Sigma,E)$ a variety in~$\mathcal{E}$, and~$(Q, \otimes, k, \leq)$ an internal partially ordered commutative monoid.
  The category~\spanvq{(\Sigma,E)}{Q} is a preordered symmetric monoidal category.
\end{restatable}
The orders are respected by the important converse operation
\begin{proposition}
  Let~$\mathcal{E}$ be a topos, and~$(\Sigma,E)$ a variety in~$\mathcal{E}$.
  Converses respect order structure, in that
  \begin{itemize}
  \item If~$(Q, \otimes, k, \bigvee)$ is an internal quantale, the converse functor of proposition~\ref{prop:relationconverse} is a partially ordered functor
  \item If~$(Q, \otimes, k, \leq)$ is an internal partially order monoid, the converse functor of proposition~\ref{prop:spanconverse} is a preordered functor
  \end{itemize}
\end{proposition}
The order enrichment of $Q$-relations and $Q$-spans is crucial for us to be able to consider the internal monads
central to the second example of the introduction.
\begin{example}
  \label{ex:metric}
  The model incorporating metric spaces as internal monads, as discussed in the introduction, can be constructed with base topos~\cset, the empty
  algebraic signature and using the Lawvere quantale~\quantale{C} as the choice of truth values.
\end{example}
Now we have both algebraic and order structure available to us within the same construction, we can consider combining the features we
are interested in, by making appropriate choices for the parameters used in the construction.
\begin{example}
  \label{ex:convexmetric}
  We now see that we can combine both the convex and metric features in a single model. With underlying topos~\cset, we take our algebraic
  structure as in example~\ref{ex:convex} and our quantale~\quantale{C} as in example~\ref{ex:metric}. In this case we find the internal
  monads are distance measures~$d : A \times A \rightarrow [0,\infty]$ satisfying
  \begin{align*}
    d(a,a) &= 0 \\
    d(a,b) + d(b,c) &\geq d(a,c) \\
    d(a_1,a_2) + d(b_1,b_2) &\geq d(a_1 +^p b_1, a_2 +^p b_2)
  \end{align*}
  These are generalized metric spaces that respect convex structure. The usual metric on~$\mathbb{R}^n$ is an example of such a metric.
\end{example}

\section{From Spans to Relations}
\label{sec:fromspanstorelations}
We now begin our study of the relationship between the different parameter choices we can take. We start with the simplest case,
the binary choice between proof relevance and provability.

The next theorem shows that the orders on relations and spans are compatible, in the sense that we can collapse spans to relations
using the join of a quantale to choose optimal truth values, and this mapping is functorial and respects the order structure.
\begin{restatable}{theorem}{vfunctor}
  \label{thm:vfunctor}
  Let~$\mathcal{E}$ be a topos, $(\Sigma,E)$ a variety in~$\mathcal{E}$ and~$(Q, \otimes, k, \bigvee)$ an internal commutative quantale.
  There is an identity on objects, strict symmetric monoidal \preord-functor
  \begin{equation*}
    V : \spanvq{(\Sigma,E)}{Q} \rightarrow \relvq{(\Sigma,E)}{Q}
  \end{equation*}
\end{restatable}
As we would expect, this extensional collapse of proof witnesses interacts well with the graph and converse operations, and therefore
preserves our chosen hypergraph structure on the nose.
\begin{proposition}
  With the same assumptions, the functor~$V$ of theorem~\ref{thm:vfunctor} commutes with graphs and converses. That is, the following
  diagrams commute:
  \begin{equation*}
    \begin{tikzpicture}[scale=0.5, node distance=2cm, ->]
      \node (bot) {\sigmaalg{\Sigma}{E}};
      \node[above left of=bot] (tl) {\spanvq{(\Sigma,E)}{Q}};
      \node[above right of=bot] (tr) {\relvq{(\Sigma,E)}{Q}};
      \draw (bot) to node[below left]{\graph{(-)}} (tl);
      \draw (bot) to node[below right]{\graph{(-)}} (tr);
      \draw (tl) to node[above]{$V$} (tr);
    \end{tikzpicture}
  \end{equation*}
  \begin{equation*}
    \begin{tikzpicture}[scale=0.5, node distance=1.5cm, ->]
      \node (tl) {$\spanvq{(\Sigma,E)}{Q}^{op}$};
      \node[below of=tl] (bl) {\spanvq{(\Sigma,E)}{Q}};
      \node[right of=tl, node distance=4cm] (tr) {$\relvq{(\Sigma,E)}{Q}^{op}$};
      \node[below of=tr] (br) {\relvq{(\Sigma,E)}{Q}};
      \draw (tl) to node[left]{\converse{(-)}} (bl);
      \draw (tr) to node[right]{\converse{(-)}} (br);
      \draw (tl) to node[above]{$V^{op}$} (tr);
      \draw (bl) to node[below]{$V$} (br);
    \end{tikzpicture}
  \end{equation*}
\end{proposition}

\section{Changing Truth Values}
\label{sec:changingtruthvalues}
We would expect that homomorphisms between our structures of truth values lead to functorial relationships
between models. This all goes through very smoothly, as we now elaborate.

Firstly, for algebraic~$Q$-relations, it is natural to consider internal quantale homomorphisms.
\begin{restatable}{theorem}{quantalemorphisminduced}
  \label{thm:quantalemorphisminduced}
  Let~$\mathcal{E}$ be a topos, $(\Sigma,E)$ a variety in~$\mathcal{E}$, and~$h : Q_1 \rightarrow Q_2$ a morphism of internal commutative quantales.
  There is an identity on objects, strict symmetric monoidal \pos-functor
  \begin{equation*}
    h^* : \relvq{(\Sigma,E)}{Q_1} \rightarrow \relvq{(\Sigma,E)}{Q_2}
  \end{equation*}
  The assignment~$h \mapsto h^*$ is functorial.
\end{restatable}
As with the extensional collapse functor of section~\ref{sec:fromspanstorelations}, the induced functor respects graphs and converses, and therefore
preserves the hypergraph structure on the nose.
\begin{proposition}
  With the same assumptions, the functor~$h^*$ of theorem~\ref{thm:quantalemorphisminduced} commutes with graphs and converses.
  That is, the following diagrams commute:
  \begin{equation*}
    \begin{tikzpicture}[scale=0.5, node distance=2cm, ->]
      \node (bot) {\sigmaalg{\Sigma}{E}};
      \node[above left of=bot] (tl) {\relvq{(\Sigma,E)}{Q_1}};
      \node[above right of=bot] (tr) {\relvq{(\Sigma,E)}{Q_2}};
      \draw (bot) to node[below left]{\graph{(-)}} (tl);
      \draw (bot) to node[below right]{\graph{(-)}} (tr);
      \draw (tl) to node[above]{$h^*$} (tr);
    \end{tikzpicture}
  \end{equation*}
  \begin{equation*}
    \begin{tikzpicture}[scale=0.5, node distance=1.5cm, ->]
      \node (tl) {$\relvq{(\Sigma,E)}{Q_1}^{op}$};
      \node[below of=tl] (bl) {\relvq{(\Sigma,E)}{Q_1}};
      \node[right of=tl, node distance=4cm] (tr) {$\relvq{(\Sigma,E)}{Q_2}^{op}$};
      \node[below of=tr] (br) {\relvq{(\Sigma,E)}{Q_2}};
      \draw (tl) to node[left]{\converse{(-)}} (bl);
      \draw (tr) to node[right]{\converse{(-)}} (br);
      \draw (tl) to node[above]{$(h^*)^{op}$} (tr);
      \draw (bl) to node[below]{$h^*$} (br);
    \end{tikzpicture}
  \end{equation*}
\end{proposition}
In the case of the span constructions, morphisms of partially ordered monoids are the appropriate notion of homomorphism to consider.
\begin{restatable}{theorem}{preorderedmonoidmorphisminduced}
  \label{thm:preorderedmonoidmorphisminduced}
  Let~$\mathcal{E}$ be a topos, $(\Sigma,E)$ a variety in~$\mathcal{E}$, and~$h : Q_1 \rightarrow Q_2$ a morphism of internal partially ordered commutative monoids.
  There is an identity on objects, strict symmetric monoidal \preord-functor
  \begin{equation*}
    h^* : \spanvq{(\Sigma,E)}{Q_1} \rightarrow \spanvq{(\Sigma,E)}{Q_2}
  \end{equation*}
  The assignment~$h \mapsto h^*$ is functorial.
\end{restatable}
Again, the induced functor commutes with graphs and converses.
\begin{proposition}
  With the same assumptions, the functor~$h^*$ of theorem~\ref{thm:preorderedmonoidmorphisminduced} commutes with graphs and converses.
  That is, the following diagrams commute:
  \begin{equation*}    
    \begin{tikzpicture}[scale=0.5, node distance=2cm, ->]
      \node (bot) {\sigmaalg{\Sigma}{E}};
      \node[above left of=bot] (tl) {\spanvq{(\Sigma,E)}{Q_1}};
      \node[above right of=bot] (tr) {\spanvq{(\Sigma,E)}{Q_2}};
      \draw (bot) to node[below left]{\graph{(-)}} (tl);
      \draw (bot) to node[below right]{\graph{(-)}} (tr);
      \draw (tl) to node[above]{$h^*$} (tr);
    \end{tikzpicture}
  \end{equation*}
  \begin{equation*}
    \begin{tikzpicture}[scale=0.5, node distance=1.5cm, ->]
      \node (tl) {$\spanvq{(\Sigma,E)}{Q_1}^{op}$};
      \node[below of=tl] (bl) {\spanvq{(\Sigma,E)}{Q_1}};
      \node[right of=tl, node distance=4cm] (tr) {$\spanvq{(\Sigma,E)}{Q_2}^{op}$};
      \node[below of=tr] (br) {\spanvq{(\Sigma,E)}{Q_2}};
      \draw (tl) to node[left]{\converse{(-)}} (bl);
      \draw (tr) to node[right]{\converse{(-)}} (br);
      \draw (tl) to node[above]{$(h^*)^{op}$} (tr);
      \draw (bl) to node[below]{$h^*$} (br);
    \end{tikzpicture}
  \end{equation*}
\end{proposition}

\begin{example}
  For any commutative quantale~$Q$ there is a partially ordered monoid morphism~$1 \rightarrow Q$, induced by the monoid unit.
  Here, $1$~is the terminal quantale. Therefore there is a strict symmetric monoidal functor
  \begin{equation*}
    \spanvq{(\Sigma,E)}{1} \rightarrow \spanvq{(\Sigma,E)}{Q}
  \end{equation*}
  This example motivates our use of partially ordered monoids, rather than simply restricting
  to the quantales of interest in our primary applications, as the required morphism is not a quantale morphism.
\end{example}

\begin{example}
  There is a quantale morphism~$\quantale{B} \rightarrow \quantale{C}$ from the Boolean to the Lawvere quantale. The induced functor identifies
  the ordinary binary relations as living within the category~\relq{C} that we introduced to capture
  metric spaces as internal monads.
\end{example}

\section{Algebraic Structure}
\label{sec:algebraicstructure}
We now investigate the interaction between truth values and algebraic structure. Again, this will lead to
functorial relationships between models, but the subject is more delicate than in the previous sections.
The essential detail is that inequation~\eqref{eq:crucialinequation} is only required to hold for the operations in our signature.
It does not directly say anything about derived terms and operations. We will require several definitions
in order to make the situation precise.
\begin{definition}
  Let~$(\Sigma,E)$ be an algebraic signature. We say that a term~$\tau$ over a finite set of variables is
  \begin{itemize}
  \item \define{Linear} if it uses each variable is used exactly once
  \item \define{Affine} if it uses each variable at most once
  \item \define{Relevant} if it uses each variable at least once
  \item \define{Cartesian} to emphasize that its use of variables is unrestricted
  \end{itemize}
  We use the same terminology for the derived operation associated to~$\tau$. 
\end{definition}
\begin{definition}[Interpretation]
  An~\define{interpretation} of signature~$(\Sigma_1,E_1)$ in signature~$(\Sigma_2,E_2)$ is a mapping
  assigning each~$\sigma \in \Sigma_1$ to a derived term of~$(\Sigma_2,E_2)$ of the same arity, such
  that the equations~$E_1$ can be proved in equational logic from~$E_2$. We say that an interpretation
  is linear (affine, relevant, cartesian) if all the derived terms used in the interpretation are
  linear (affine, relevant, cartesian). We write~\linsignatures (\affsignatures, \relsignatures, \cartsignatures)
  for the corresponding categories with objects signatures and morphisms linear (affine, relevant, cartesian) interpretations.
\end{definition}
\begin{definition}
  Let~$\mathcal{E}$ be a topos, and~$(Q, \otimes, k, \leq)$ an internal partially ordered monoid. We say that~$Q$ is
  \begin{itemize}
  \item \define{Linear} to emphasize that no additional axioms are assumed to hold
  \item \define{Affine} if the axiom
    \begin{equation*}
      \forall p, q. p \otimes q \leq p
    \end{equation*}
    is valid
  \item \define{Relevant} if the axiom
    \begin{equation*}
      \forall q. q \leq q \otimes q
    \end{equation*}
    is valid
  \item \define{Cartesian} if it is both affine and relevant
  \end{itemize}
\end{definition}
We note the following important special case.
\begin{example}
  A cartesian commutative quantale is a locale.
\end{example}
So in the case where our truth values have a genuine locale structure, everything becomes very well behaved.
\begin{definition}
  Let~$\mathcal{E}$ be a topos, and~$(Q, \otimes, k, \bigvee)$ an internal commutative quantale. We say that a $Q$-relation is
  \begin{itemize}
  \item \define{Linear} to emphasize that no additional axioms are assumed to hold
  \item \define{Affine} if the axiom
    \begin{equation*}
      \forall a_1,a_2,b_1,b_2. R(a_1,b_1) \otimes R(a_2,b_2) \leq R(a_1,b_1)
    \end{equation*}
    is valid
  \item \define{Relevant} if the axiom
    \begin{equation*}
      \forall a,b. R(a,b) \leq R(a,b) \otimes R(a,b)
    \end{equation*}
    is valid
  \item \define{Cartesian} if it is both affine and relevant
  \end{itemize}
  We write \linrelvq{(\Sigma,E)}{Q}, \affrelvq{(\Sigma,E)}{Q}, \relrelvq{(\Sigma,E)}{Q} and \cartrelvq{(\Sigma,E)}{Q}
  for the corresponding subcategories of algebraic Q-relations.
\end{definition}
Our terminology is derived from that sometimes used for variants of linear logic. If we view truth values as resources, the question is when
can these resources be ``copied'' or ``deleted''. We have adopted a naming convention that is slightly redundant, for example~$\linrelvq{(\Sigma,E)}{Q}$ is the same thing as~$\relvq{(\Sigma,E)}{Q}$.
We permit this redundancy in order to allow uniform statements of the subsequent theorems. We begin with important closure properties of our various classes of morphisms.
\begin{proposition}
  \label{prop:relclosure}
  The subcategories of linear (affine, relevant, cartesian) algebraic $Q$-relations are closed under tensors, converses and the functors induced
  by quantale homomorphisms. Also, the algebraic $Q$-relations in the image of the graph functor are all cartesian.
\end{proposition}
A straightforward corollary of the closure properties of proposition~\ref{prop:relclosure} is
\begin{theorem}
  For a topos~$\mathcal{E}$, variety ~$(\Sigma,E)$ in~$\mathcal{E}$ and internal commutative quantale~$(Q, \otimes, k, \bigvee)$, the categories~\linrelvq{(\Sigma,E)}{Q}, \affrelvq{(\Sigma,E)}{Q},
  \relrelvq{(\Sigma,E)}{Q} and \cartrelvq{(\Sigma,E)}{Q} are sub-hypergraph categories of~\relvq{(\Sigma,E)}{Q}.
\end{theorem}
Our restricted classes of relations respect the corresponding classes of terms.
\begin{restatable}{proposition}{relclosedunderterms}\label{prop:relclosedunderterms}
  Let~$\mathcal{E}$ be a topos, $(\Sigma,E)$ a variety in~$\mathcal{E}$ and~$(Q, \otimes, k, \bigvee)$ an internal commutative quantale. For linear (affine, relevant, cartesian)
  algebraic $Q$-relation~$R : A \rightarrow B$ the axiom
  \begin{equation*}
    R(a_1,b_1) \otimes ... \otimes R(a_n,b_n) \leq R(\tau(a_1,...,a_n),\tau(b_1,...,b_n))
  \end{equation*}
  holds for every linear (affine, relevant, cartesian)~$n$-ary derived operation~$\tau$.
\end{restatable}
The next proposition is straightforward, it establishes that once our truth values are sufficiently nice, all our relations inherit the same property.
\begin{restatable}{proposition}{relniceness}
  \label{prop:relniceness}
  Let~$\mathcal{E}$ be a topos, $(\Sigma,E)$ a variety in~$\mathcal{E}$ and~$(Q, \otimes, k, \bigvee)$ an internal commutative quantale. If~$Q$ is linear (affine, relevant, cartesian),
  every morphism in~\relvq{(\Sigma,E)}{Q} is linear (affine, relevant, cartesian).
\end{restatable}
In particular, if our quantale is in fact a locale, proposition~\ref{prop:relniceness} tells us that everything becomes as straightforward as we might hope.

Finally, we can establish a contravariant functorial relationship between interpretations and functors between models.
\begin{restatable}{theorem}{interpretationinducedrel}
  \label{thm:interpretationinducedrel}
  Let~$\mathcal{E}$ be a topos and~$(Q, \otimes, k, \bigvee)$ an internal commutative quantale.
  Let~$i : (\Sigma_1, E_1) \rightarrow (\Sigma_2, E_2)$ be a linear interpretation of signatures. There is a strict monoidal functor
  \begin{equation*}
    i^* : \linrelvq{(\Sigma_2,E_2)}{Q} \rightarrow \linrelvq{(\Sigma_1,E_1)}{Q}
  \end{equation*}
  The assignment~$i \mapsto i^*$ extends to a contravariant functor.

  Similar results hold for affine, relevant and cartesian interpretations and relations.
\end{restatable}
As usual, the induced functors respect the essential graph and converse structure.
\begin{proposition}
  With the same assumptions, the induced functor~$i^*$ of theorem~\ref{thm:interpretationinducedrel} commutes with graphs and converses.
  That is, the following diagrams commute:
  \begin{equation*}
    \begin{tikzpicture}[scale=0.5, node distance=1.5cm, ->]
      \node (bl) {\sigmaalg{\Sigma_2}{E_2}};
      \node[above of=bl] (tl) {\linrelvq{(\Sigma_2,E_2)}{Q}};
      \node[right of=bl, node distance=4cm] (br) {\sigmaalg{\Sigma_1}{E_1}};
      \node[above of=br] (tr) {\linrelvq{(\Sigma_1,E_1)}{Q}};
      \draw (bl) to node[left]{\graph{(-)}} (tl);
      \draw (br) to node[right]{\graph{(-)}} (tr);
      \draw (bl) to node[below]{$i^*$} (br);
      \draw (tl) to node[above]{$i^*$} (tr);
    \end{tikzpicture}
  \end{equation*}
  \begin{equation*}
    \begin{tikzpicture}[scale=0.5, node distance=1.5cm, ->]
      \node (tl) {$\linrelvq{(\Sigma_2,E_2)}{Q}^{op}$};
      \node[below of=tl] (bl) {\linrelvq{(\Sigma_2,E_2)}{Q}};
      \node[right of=tl, node distance=4cm] (tr) {$\linrelvq{(\Sigma_1,E_1)}{Q}^{op}$};
      \node[below of=tr] (br) {\linrelvq{(\Sigma_1,E_1)}{Q}};
      \draw (tl) to node[left]{\converse{(-)}} (bl);
      \draw (tr) to node[right]{\converse{(-)}} (br);
      \draw (tl) to node[above]{$(i^*)^{op}$} (tr);
      \draw (bl) to node[below]{$i^*$} (br);
    \end{tikzpicture}
  \end{equation*}
  The bottom functor in these diagrams is the obvious induced functor between categories of algebras.
  Similar diagrams commute for affine, relevant and cartesian interpretations and relations.
\end{proposition}
We now introduce similar definitions for the setting of spans, in order to proceed with a similar analysis.
\begin{definition}
  Let~$\mathcal{E}$ be a topos, and~$(Q, \otimes, k, \leq)$ an internal partially ordered monoid. We say that a~$Q$-span~$(X,f,g,\chi)$ is
  \begin{itemize}
  \item \define{Linear} to emphasize that no additional axioms are assumed to hold
  \item \define{Affine} if the axiom~$\forall x_1,x_2. \chi(x_1) \otimes \chi(x_2) \leq \chi(x_1)$ is valid
  \item \define{Relevant} if the axiom~$\forall x. \chi(x) \leq \chi(x) \otimes \chi(x)$ is valid
  \item \define{Cartesian} if it is both affine and relevant
  \end{itemize}
  We write \linspanvq{(\Sigma,E)}{Q}, \affspanvq{(\Sigma,E)}{Q}, \relspanvq{(\Sigma,E)}{Q} and \cartspanvq{(\Sigma,E)}{Q}
  for the corresponding subcategories of algebraic Q-spans.
\end{definition}
Spans with sufficient structure respect the corresponding types of terms.
\begin{restatable}{proposition}{spanclosedterms}\label{prop:spanclosedunderterms}
  Let~$\mathcal{E}$ be a topos, $(\Sigma,E)$ a variety in~$\mathcal{E}$, and~$(Q, \otimes, k, \leq)$ an internal partially ordered commutative monoid.
  For~$(\Sigma,E)$-algebras~$A$ and~$B$, and linear (affine, relevant, cartesian) algebraic $Q$-span~$(X,f,g,\chi)$
  and $n$-ary linear (affine, relevant, cartesian) term~$\tau$ if
  \begin{equation*}
    f(x_1) = a_1 \wedge g(x_1) = b_1 \wedge ... \wedge f(x_n) = a_n \wedge g(x_n) = b_n
  \end{equation*}
  then there exists~$x$ such that
  \begin{equation*}
    f(x) = \tau(a_1,...,a_n) \wedge g(x) = \tau(b_1,...,b_n)
  \end{equation*}
  and
  \begin{equation*}
    \chi(x_1) \otimes ... \otimes \chi(x_n) \leq \chi(x)
  \end{equation*}
\end{restatable}
Again, we have good closure of our various classes of morphisms.
\begin{proposition}
  \label{prop:spanclosure}
  The subcategories of linear (affine, relevant, cartesian) algebraic $Q$-spans are closed under tensors, converses and the functors induced
  by quantale homomorphisms. Also, the algebraic $Q$-spans in the image of the graph functor are all cartesian.
\end{proposition}
As with relations, the closure properties of proposition~\ref{prop:spanclosure} yield a straightforward corollary about our subcategories of algebraic~$Q$-spans.
\begin{theorem}
  For a topos~$\mathcal{E}$, variety ~$(\Sigma,E)$ in~$\mathcal{E}$ and internal partially ordered monoid~$(Q, \otimes, k, \leq)$, the categories~\linspanvq{(\Sigma,E)}{Q}, \affspanvq{(\Sigma,E)}{Q},
  \relspanvq{(\Sigma,E)}{Q} and \cartspanvq{(\Sigma,E)}{Q} are sub-hypergraph categories of~\spanvq{(\Sigma,E)}{Q}.
\end{theorem}
As with relations, algebraic~$Q$-spans inherit good properties from their underlying quantale.
\begin{proposition}
  Let~$\mathcal{E}$ be a topos, $(\Sigma,E)$ a variety in~$\mathcal{E}$ and~$(Q, \otimes, k, \leq)$ an internal partially ordered monoid. If~$Q$ is linear (affine, relevant, cartesian),
  every morphism in~\spanvq{(\Sigma,E)}{Q} is linear (affine, relevant, cartesian).
\end{proposition}
Again, we can now establish a contravariant functorial relationship between interpretations and functors between models.
\begin{restatable}{theorem}{interpretationinducedspan}
  \label{thm:interpretationinducedspan}
  Let~$\mathcal{E}$ be a topos and~$(Q, \otimes, k, \leq)$ an internal partially ordered commutative monoid.
  Let~$i : (\Sigma_1, E_1) \rightarrow (\Sigma_2, E_2)$ be a linear interpretation of signatures. There is a strict monoidal functor
  \begin{equation*}
    i^* : \linspanvq{(\Sigma_2,E_2)}{Q} \rightarrow \linspanvq{(\Sigma_1,E_1)}{Q}
  \end{equation*}
  The assignment~$i \mapsto i^*$ extends to a contravariant functor.

  Similar results hold for affine, relevant and cartesian interpretations and spans.
\end{restatable}
The induced functors respect the usual essential structure.
\begin{proposition}
  For the same assumptions, the induced functor~$i^*$ of theorem~\ref{thm:interpretationinducedspan} commutes with
  graphs and converses. That is, the following diagrams commute
  \begin{equation*}
    \begin{tikzpicture}[scale=0.5, node distance=1.5cm, ->]
      \node (bl) {\sigmaalg{\Sigma_2}{E_2}};
      \node[above of=bl] (tl) {\linspanvq{(\Sigma_2,E_2)}{Q}};
      \node[right of=bl, node distance=4cm] (br) {\sigmaalg{\Sigma_1}{E_1}};
      \node[above of=br] (tr) {\linspanvq{(\Sigma_1,E_1)}{Q}};
      \draw (bl) to node[left]{\graph{(-)}} (tl);
      \draw (br) to node[right]{\graph{(-)}} (tr);
      \draw (bl) to node[below]{$i^*$} (br);
      \draw (tl) to node[above]{$i^*$} (tr);
    \end{tikzpicture}
  \end{equation*}
  \begin{equation*}
    \begin{tikzpicture}[scale=0.5, node distance=1.5cm, ->]
      \node (tl) {$\linspanvq{(\Sigma_2,E_2)}{Q}^{op}$};
      \node[below of=tl] (bl) {\linspanvq{(\Sigma_2,E_2)}{Q}};
      \node[right of=tl, node distance=4cm] (tr) {$\linspanvq{(\Sigma_1,E_1)}{Q}^{op}$};
      \node[below of=tr] (br) {\linspanvq{(\Sigma_1,E_1)}{Q}};
      \draw (tl) to node[left]{\converse{(-)}} (bl);
      \draw (tr) to node[right]{\converse{(-)}} (br);
      \draw (tl) to node[above]{$(i^*)^{op}$} (tr);
      \draw (bl) to node[below]{$i^*$} (br);
    \end{tikzpicture}
  \end{equation*}
  The bottom functor in these diagrams is the obvious induced functor between categories of algebras.
  Similar diagrams commute for affine, relevant and cartesian interpretations and relations.
\end{proposition}
The extensional collapse functor of section~\ref{sec:fromspanstorelations} also respects our different classes of spans
and relations.
\begin{proposition}
  Let~$\mathcal{E}$ be a topos, $(\Sigma,E)$ a variety in~$\mathcal{E}$ and~$(Q, \otimes, k, \bigvee)$ an internal commutative quantale.
  The functor of theorem~\ref{thm:vfunctor} maps linear (affine, relevant, cartesian) algebraic \mbox{$Q$-spans}
  to linear (affine, relevant, cartesian) algebraic \mbox{$Q$-relations}.
\end{proposition}
We briefly discuss some examples.
\begin{example}
  Let~$(\emptyset, \emptyset)$ denote the signature with no operations or equations. For any signature~$(\Sigma,E)$ there is a trivial \emph{linear}
  interpretation~$(\emptyset, \emptyset) \rightarrow (\Sigma,E)$.
  We therefore have, for every choice of internal quantale~$Q$, strict symmetric monoidal forgetful functors
  \begin{align*}
    \relvq{(\Sigma,E)}{Q} &\rightarrow \relvq{(\emptyset, \emptyset)}{Q}\\
    \spanvq{(\Sigma,E)}{Q} &\rightarrow \spanvq{(\emptyset, \emptyset)}{Q}
  \end{align*}
\end{example}
\begin{example}
  We can present real vector spaces by a signature with a constant element representing the origin, and a family
  of binary mixing operations, indexed by the scalars involved, satisfying suitable equations. We denote this signature as~\signature{Linear}. There is an interpretation
  in~\linsignatures of type~$\signature{Convex} \rightarrow \signature{Linear}$. For any commutative quantale~$Q$, this interpretation induces a functor
  \begin{equation*}
    \relvq{\signature{Linear}}{Q} \rightarrow \relvq{\signature{Convex}}{Q}
  \end{equation*}
  So, as we would expect, we can find the vector spaces in the convex algebras, in a manner respecting all the relevant categorical structure.
\end{example}
\begin{example}
  An affine join semilattice is a set with an associative, commutative, idempotent binary operation. From an information theoretic perspective, we think of convex algebras as describing probabilistic
  ambiguity. Affine join semilattices can then be thought of as modelling unquantified ambiguity. If we denote the signature for affine join semilattices as~$\signature{Affine}$
  there is an interpretation in~\linsignatures of type~$\signature{Convex} \rightarrow \signature{Affine}$ inducing a functor
  \begin{equation*}
    \relvq{\signature{Affine}}{Q} \rightarrow \relvq{\signature{Convex}}{Q}
  \end{equation*}
  relating these two different models of epistemic phenomena. This exhibits another interesting subcategory of~\convexrel.
\end{example}

\section{Changing Topos}
\label{sec:changingtopos}
We now explore the last axis of variation, the topos structure. 
We would expect that, if~$\mathcal{E}$ and~$\mathcal{F}$ are elementary toposes,
given a suitable functor~$L:\mathcal{E} \to \mathcal{F}$ it would be possible to lift it to a functor between their respective
relation and span constructions. Since the definitions of these categories make wide use of the internal language, it should not be surprising that
by ``suitable'' we actually mean that~$L$ behaves well with respect to the logical properties of $\mathcal{E}, \mathcal{F}$.
\begin{definition}
	Given toposes $\mathcal{E}, \mathcal{F}$, a functor $L:\mathcal{E} \to \mathcal{F}$ is called \emph{logical} if:
	\begin{itemize}
		\item $L$ preserves products
		\item $L$ preserves exponentials
		\item $L$ preserves the subobject classifier.
	\end{itemize}
\end{definition}
Logical functors are the right functors to consider, 
since they preserve the validity of internal formulas: If $\models \phi$ in $\mathcal{E}$, then $\models \phi$ in $\mathcal{F}$ for 
every formula $\phi$ written in the language of first order intuitionistic logic.

To make the following results more readable, we will have to slightly refine our notation, 
writing~$\relvtq{\mathcal{E}}{(\Sigma, E)}{Q}$  and~$\spanvtq{\mathcal{E}}{\Sigma, E}{Q}$ to explicitly indicate
that the constructions are performed on topos~$\mathcal{E}$. 
If $L:\mathcal{E} \rightarrow \mathcal{F}$ is a logical functor and $Q$ 
is an internal quantale in $\mathcal{E}$, then the fact that $L$ preserves models 
of first order intuitionistic theories implies that $LQ$ is an internal quantale 
in $\mathcal{F}$. It makes sense, then, to consider how 
$\relvtq{\mathcal{E}}{(\Sigma, E)}{Q}$ and $\relvtq{\mathcal{F}}{(\Sigma, E)}{LQ}$ 
are related. The main result of the section is the following:

\begin{restatable}{theorem}{logicalfunctorinducedrel}
	\label{thm:logicalfunctorinducedrel}
	Let~$\mathcal{E}, \mathcal{F}$ be toposes, and $L:\mathcal{E} \rightarrow \mathcal{F}$ be a logical functor.
	Let~$(Q, \otimes, k, \bigvee)$ be an internal commutative quantale in $\mathcal{E}$ and~$(\Sigma, E)$ be a signature.
	There is a symmetric monoidal functor
	\begin{equation*}
		L^*: \relvtq{\mathcal{E}}{(\Sigma, E)}{Q} \rightarrow \relvtq{\mathcal{F}}{(\Sigma, E)}{LQ}
	\end{equation*}
	The assignment~$L \mapsto L^*$ is functorial.
	
\end{restatable}
As in the previous cases, graph and converse functors are preserved.
\begin{restatable}{proposition}{logicalcommutesrel}
	With the same assumptions, the induced functor~$L^*$ of theorem~\ref{thm:logicalfunctorinducedrel} commutes with graphs and converses.
	That is, the following diagrams commute:
        
	\begin{equation*}
	\begin{tikzpicture}[scale=0.5, node distance=1.5cm, ->]
	\node (bl) {$\sigmaalgt{\mathcal{E}}{\Sigma}{E}$};
	\node[above of=bl] (tl) {\relvtq{\mathcal{E}}{(\Sigma, E)}{Q}};
	\node[right of=bl, node distance=4cm] (br) {$\sigmaalgt{\mathcal{F}}{\Sigma}{E}$};
	\node[above of=br] (tr) {\relvtq{\mathcal{F}}{(\Sigma, E)}{LQ}};
	\draw (bl) to node[left]{\graph{(-)}} (tl);
	\draw (br) to node[right]{\graph{(-)}} (tr);
	\draw (bl) to node[below]{$L$} (br);
	\draw (tl) to node[above]{$L^*$} (tr);
	\end{tikzpicture}
	\end{equation*}
	\begin{equation*}
	\begin{tikzpicture}[scale=0.5, node distance=1.5cm, ->]
	\node (tl) {$\relvtq{\mathcal{E}}{(\Sigma, E)}{Q}^{op}$};
	\node[below of=tl] (bl) {\relvtq{\mathcal{E}}{(\Sigma, E)}{Q}};
	\node[right of=tl, node distance=4cm] (tr) {$\relvtq{\mathcal{F}}{(\Sigma, E)}{LQ}^{op}$};
	\node[below of=tr] (br) {\relvtq{\mathcal{F}}{(\Sigma, E)}{LQ}};
	\draw (tl) to node[left]{\converse{(-)}} (bl);
	\draw (tr) to node[right]{\converse{(-)}} (br);
	\draw (tl) to node[above]{$(L^*)^{op}$} (tr);
	\draw (bl) to node[below]{$L^*$} (br);
	\end{tikzpicture}
	\end{equation*}
\end{restatable}	
As with relations, morphisms between toposes extend functorially to morphisms between spans.	
\begin{restatable}{theorem}{logicalfunctorinducedspan}
	\label{thm:logicalfunctorinducedspan}
	Let~$\mathcal{E}, \mathcal{F}$ be toposes, and $L:\mathcal{E} \rightarrow \mathcal{F}$ be a logical functor.
	Let~$(Q, \otimes, k, \leq)$ be an internal partially ordered commutative monoid in $\mathcal{E}$ and~$(\Sigma, E)$ be a signature.
	There is a symmetric monoidal functor
	\begin{equation*}
	L^*: \spanvtq{\mathcal{E}}{(\Sigma, E)}{Q} \rightarrow \spanvtq{\mathcal{F}}{(\Sigma, E)}{LQ}
	\end{equation*}
	The assignment~$L \mapsto L^*$ is functorial.	
\end{restatable}
The essential structure is again respected by the induced functors.
\begin{proposition}
	With the same assumptions, the induced functor~$L^*$ of theorem~\ref{thm:logicalfunctorinducedspan} commutes with graphs and converses.
	That is, the following diagrams commute:
	\begin{equation*}
	\begin{tikzpicture}[scale=0.5, node distance=1.5cm, ->]
	\node (bl) {$\sigmaalgt{\mathcal{E}}{\Sigma}{E}$};
	\node[above of=bl] (tl) {\spanvtq{\mathcal{E}}{(\Sigma, E)}{Q}};
	\node[right of=bl, node distance=4cm] (br) {$\sigmaalgt{\mathcal{F}}{\Sigma}{E}$};
	\node[above of=br] (tr) {\spanvtq{\mathcal{F}}{(\Sigma, E)}{LQ}};
	\draw (bl) to node[left]{\graph{(-)}} (tl);
	\draw (br) to node[right]{\graph{(-)}} (tr);
	\draw (bl) to node[below]{$L$} (br);
	\draw (tl) to node[above]{$L^*$} (tr);
	\end{tikzpicture}
	\end{equation*}
	\begin{equation*}
	\begin{tikzpicture}[scale=0.5, node distance=1.5cm, ->]
	\node (tl) {$\spanvtq{\mathcal{E}}{(\Sigma, E)}{Q}^{op}$};
	\node[below of=tl] (bl) {\spanvtq{\mathcal{E}}{(\Sigma, E)}{Q}};
	\node[right of=tl, node distance=4cm] (tr) {$\spanvtq{\mathcal{F}}{(\Sigma, E)}{LQ}^{op}$};
	\node[below of=tr] (br) {\spanvtq{\mathcal{F}}{(\Sigma, E)}{LQ}};
	\draw (tl) to node[left]{\converse{(-)}} (bl);
	\draw (tr) to node[right]{\converse{(-)}} (br);
	\draw (tl) to node[above]{$(L^*)^{op}$} (tr);
	\draw (bl) to node[below]{$L^*$} (br);
	\end{tikzpicture}
	\end{equation*}
\end{proposition}	
\begin{example}
  Given any category $\mathcal{C}$ we can form a corresponding \emph{presheaf category},
  having representable functors from $\mathcal{C}$ to $\cset$ as objects and natural transformations between them
  as morphisms. Presheaves constitute one of the most important examples of toposes, and it makes sense to ask how 
  Theorems~\ref{thm:logicalfunctorinducedrel}, \ref{thm:logicalfunctorinducedspan} behave in these circumstances.
  
  In general, given arbitrary categories~$\mathcal{C}, \mathcal{D}$ it is difficult to say when a 
  functor $F:\mathcal{C} \to \mathcal{D}$ lifts to a logical functor between the corresponding presheaves. 
  Nevertheless, the following result holds: If $\mathcal{C}, \mathcal{D}$ are groupoids (categories in which every arrow is an isomorphism), 
  then any functor $F:\mathcal{C} \to \mathcal{D}$ lifts to a logical functor $\bar{F}: \mathcal{C} \to \mathcal{D}$. This is because
  truth values in presheaf toposes are defined in terms of \emph{sieves} (subfunctors of the homset functor) and these sieves trivialize
  when the only arrows at our disposal are isos. This in turn trivializes the structure of truth values in the presheaf itself, 
  that ends up to be defined pointwise from $\cset$.
  
  Theorems~\ref{thm:logicalfunctorinducedrel}, \ref{thm:logicalfunctorinducedspan} 
  then ensure that~$\bar{F}$ can be lifted to the relational and span structures 
  built on~$\cset^{\mathcal{C}}$ and~$\cset^\mathcal{D}$, respectively.

\end{example}	
\begin{example}
	If~$\mathcal{E}$ is a topos, and~$f:I \to J$ is a morphism of~$\mathcal{E}$, then pulling back along~$f$ induces a logical functor~$F:\mathcal{E}/J \to\mathcal{E}/I$. Theorem~\ref{thm:logicalfunctorinducedrel} guarantees the existence of a functor~$F^*: \relvtq{\mathcal{E}/J}{(\Sigma, E)}{Q} \to \relvtq{\mathcal{E}/I}{(\Sigma, E)}{FQ}$. In particular, this means that there is always a functor~$F^*: \relvtq{\mathcal{E}}{(\Sigma, E)}{Q} \to \relvtq{\mathcal{E}/I}{(\Sigma, E)}{FQ}$, where~$\mathcal{E}/I$ is any slice topos of~$\mathcal{E}$.
\end{example}

\section{Independence of the axes of variation}
Finally, we establish that our various induced functors between models are independent,
in that they all commute with each other.Unfortunately, the commutativity of the functors induced by interpretations between algebras, order structure and quantale morphisms with $L^*$ will hold only up to isomorphism.
This depends intrinsicly on the definition of logical functor, that is, in turn, defined to preserve validity of formulas in the internal language only up to natural isomorphism.

\begin{restatable}{theorem}{thebox}
	\label{thm:thebox}
	Let~$\mathcal{E}$ be a topos, $h : Q_1 \rightarrow Q_2$ a morphism of internal commutative quantales,
	 $i : (\Sigma_1, E_1) \rightarrow (\Sigma_2, E_2)$ a linear interpretation and $L:\mathcal{E} \to \mathcal{F}$ a logical functor. For the induced
	functors of theorems~\ref{thm:quantalemorphisminduced}, \ref{thm:preorderedmonoidmorphisminduced}, \ref{thm:interpretationinducedrel}, \ref{thm:interpretationinducedspan}, \ref{thm:logicalfunctorinducedrel} and \ref{thm:logicalfunctorinducedspan},
	the following diagram commutes (be aware that in the hypercube below commutative squares involving~$L^*$ only commute up to isomorphism. Other squares commute up to equality):
	\begin{equation*}
	\begin{tikzpicture}[scale=1.5, ->]
	\tikzstyle{vertex}=[circle,minimum size=20pt,inner sep=0pt]
	\tikzstyle{selected vertex} = [vertex, fill=red!24]
	\node[vertex] (fbl) at (-0.2,-0.2) {$\bullet$};
	\node[vertex] (ftl) at (-0.2,0.8) {$\bullet$};
	\node[vertex] (fbr) at (1,-0.2) {$\bullet$};
	\node[vertex] (ftr) at (1,0.8) {$\bullet$};
	\node[vertex] (bbl) at (0.2, 0.4) {$\bullet$};
	\node[vertex] (btl) at (0.2,1.4) {$\bullet$};
	\node[vertex] (bbr) at (1.3,0.4) {$\bullet$};
	\node[vertex] (btr) at (1.3,1.4) {$\bullet$};
	
	\node[vertex] (fblo) at (-1.3,-1.2) {$\bullet$};
	\node[vertex] (ftlo) at (-1.3,1.8) {$\bullet$};
	\node[vertex] (fbro) at (2,-1.2) {$\bullet$};
	\node[vertex] (ftro) at (2,1.8) {$\bullet$};	
	
	\node[vertex] (bblo) at (-0.5,0) {$\bullet$};
	\node[vertex] (btlo) at (-0.5,3) {$\bullet$};
	\node[vertex] (bbro) at (2.5,0) {$\bullet$};
	\node[vertex] (btro) at (2.5,3) {$\bullet$};	
	
	\draw (fbl) -- (fbr);
	\draw (ftl) -- (fbl);
	\draw (ftl) -- (ftr);
	\draw (ftr) -- (fbr);

	\draw (bbl) -- (bbr);
	\draw (btl) -- (bbl);
	\draw (btl) -- (btr);
	\draw (btr) -- (bbr);

	\draw (bbl) -- (fbl);
	\draw (bbr) -- (fbr);
	\draw (btl) -- (ftl);
	\draw (btr) -- (ftr);
	
	\draw (fblo) -- (fbro);
	\draw (ftlo) -- (fblo);
	\draw (ftlo) -- (ftro);
	\draw (ftro) -- (fbro);
	
	\draw (bblo) -- (bbro);
	\draw (btlo) -- (bblo);
	\draw (btlo) -- (btro);
	\draw (btro) -- (bbro);
	
	\draw (bblo) -- (fblo);
	\draw (bbro) -- (fbro);
	\draw (btlo) -- (ftlo);
	\draw (btro) -- (ftro);
	
	\draw (fbl) to node[below right]{$L^*$} (fblo);
	\draw (ftl) to node[below left]{$L^*$} (ftlo);
	\draw (fbr) to node[below left]{$L^*$} (fbro);
	\draw (ftr) to node[below right]{$L^*$} (ftro);
	\draw (bbl) to (bblo);
	\draw (btl) to node[right]{$L^*$} (btlo);
	\draw (bbr) to node[above]{$L^*$} (bbro);
	\draw (btr) to node[above]{$L^*$} (btro);
	\end{tikzpicture}
	\end{equation*}
	
	Where the inner cube is
	\begin{equation*}
	\begin{tikzpicture}[scale=0.4, node distance=2cm, ->]
	\node (backtl) {\linspanvtq{\mathcal{E}}{(\Sigma_2,E_2)}{Q_1}};
	\node[right of=backtl, node distance=4cm] (backtr) {\linspanvtq{\mathcal{E}}{(\Sigma_1,E_1)}{Q_1}};
	\node[below left of=backtl] (fronttl) {\linspanvtq{\mathcal{E}}{(\Sigma_2,E_2)}{Q_2}};
	\node[right of=fronttl, node distance=4cm] (fronttr) {\linspanvtq{\mathcal{E}}{(\Sigma_1, E_1)}{Q_2}};
	
	\node[below of=backtl, node distance=3cm] (backbl) {\linrelvtq{\mathcal{E}}{(\Sigma_2, E_2)}{Q_1}};
	\node[below of=backtr, node distance=3cm] (backbr) {\linrelvtq{\mathcal{E}}{(\Sigma_1, E_1)}{Q_1}};
	\node[below of=fronttl, node distance=3cm] (frontbl) {\linrelvtq{\mathcal{E}}{(\Sigma_2, E_2)}{Q_2}};
	\node[below of=fronttr, node distance=3cm] (frontbr) {\linrelvtq{\mathcal{E}}{(\Sigma_1, E_1)}{Q_2}};
	
	\draw (backtl) to node[above]{$i^*$} (backtr);
	\draw (fronttl) to node[above]{$i^*$} (fronttr);
	\draw (backtl) to node[left]{$h^*$} (fronttl);
	\draw (backtr) to node[left]{$h^*$} (fronttr);
	\draw (backbl) to node[above]{$i^*$} (backbr);
	\draw (frontbl) to node[above]{$i^*$} (frontbr);
	\draw (backbl) to node[right]{$h^*$} (frontbl);
	\draw (backbr) to node[right]{$h^*$} (frontbr);
	\draw (backtl) -- (backbl);
	\draw (fronttl) -- (frontbl);
	\draw (backtr) -- (backbr);
	\draw (fronttr) -- (frontbr);
	\end{tikzpicture}
	\end{equation*}
	and the outer cube is 
	\begin{equation*}
	\begin{tikzpicture}[scale=0.4, node distance=2cm, ->]
	\node (backtl) {\linspanvtq{\mathcal{F}}{(\Sigma_2,E_2)}{LQ_1}};
	\node[right of=backtl, node distance=4cm] (backtr) {\linspanvtq{\mathcal{F}}{(\Sigma_1,E_1)}{LQ_1}};
	\node[below left of=backtl] (fronttl) {\linspanvtq{\mathcal{F}}{(\Sigma_2,E_2)}{LQ_2}};
	\node[right of=fronttl, node distance=4cm] (fronttr) {\linspanvtq{\mathcal{F}}{(\Sigma_1, E_1)}{LQ_2}};
	
	\node[below of=backtl, node distance=3cm] (backbl) {\linrelvtq{\mathcal{F}}{(\Sigma_2, E_2)}{LQ_1}};
	\node[below of=backtr, node distance=3cm] (backbr) {\linrelvtq{\mathcal{F}}{(\Sigma_1, E_1)}{LQ_1}};
	\node[below of=fronttl, node distance=3cm] (frontbl) {\linrelvtq{\mathcal{F}}{(\Sigma_2, E_2)}{LQ_2}};
	\node[below of=fronttr, node distance=3cm] (frontbr) {\linrelvtq{\mathcal{F}}{(\Sigma_1, E_1)}{LQ_2}};
	
	\draw (backtl) to node[above]{$i^*$} (backtr);
	\draw (fronttl) to node[above]{$i^*$} (fronttr);
	\draw (backtl) to node[left]{$(Lh)^*$} (fronttl);
	\draw (backtr) to node[left]{$(Lh)^*$} (fronttr);
	\draw (backbl) to node[above]{$i^*$} (backbr);
	\draw (frontbl) to node[above]{$i^*$} (frontbr);
	\draw (backbl) to node[right]{$(Lh)^*$} (frontbl);
	\draw (backbr) to node[right]{$(Lh)^*$} (frontbr);
	\draw (backtl) -- (backbl);
	\draw (fronttl) -- (frontbl);
	\draw (backtr) -- (backbr);
	\draw (fronttr) -- (frontbr);
	\end{tikzpicture}
	\end{equation*}
	In both cases the vertical arrows are the functors of theorem~\ref{thm:vfunctor}.
	Similar diagrams commute for affine, relevant and cartesian interpretations, relations and spans.
\end{restatable}

\section{Conclusion}
We have developed a parameterized scheme for constructing hypergraph categories, by generalizing the
notion of binary relation along four axes of variation:
\begin{itemize}
\item The ambient mathematical background via the choice of underlying category
\item The truth values via a choice of internal quantale
\item The choice of algebraic structure
\item The choice between proof relevance and provability
\end{itemize}
This construction provides a conceptually motivated approach for producing models of process theories when investigating new applications.
Many existing examples in the literature are covered by our approach, including examples used for linguistics, cognition, linear dynamical
systems and non-deterministic computation.

We showed that the resulting categories are preorder enriched, providing more flexible modelling possibilities.
It was also established that varying each of the parameters is functorial,
preserving all the important hypergraph and order structure. 
In the case of the algebraic structure, this functoriality
exhibited an interesting relationship between algebra and resource sensitivity in the sense of linear logic.
Our constructions were also shown to have well behaved functorial analogues of the notions of taking the converse of a relations,
and taking the graph of a map to construct a new relation.

Interestingly, our framework points to new models in which features can be combined. This was a key objective of
this direction of research.
For example the model incorporating both convexity and metrics of example~\ref{ex:convexmetric},
the proof relevant models of cognition of example~\ref{ex:convexspan} and the possibility of incorporating contexts as discussed in example~\ref{ex:presheaf}.
The application of these constructions to models of cognition and natural language will be explored in forthcoming work.

In order to gain a strict composition operation, in section~\ref{sec:spans} we used isomorphism classes of spans,
and then introduced an analogue of the usual order structure for relations in section~\ref{sec:orderstructure}.
If we use spans, rather than their equivalence classes, they should form a symmetric monoidal bicategory, sacrificing strict composition for a richer 2-cell structure.
This is of practical interest as internal monads have been
important in our examples. The internal monads in categories of spans correspond to internal categories~\cite{RosebrughWood2002},
which would open up further interesting possibilities.
Some related work on bicategorical aspects of the decorated cospan construction appears in~\cite{Courser2016}.
In fact, the resulting categories should be an appropriate bicategorical generalization of a hypergraph category. Such bicategorical
aspects are left to later work.

\subsection*{Acknowledgments}
The authors would like to thank Bob Coecke, Ignacio Funke, Kohei Kishida and Martha Lewis for feedback and discussions.
This work was partially funded by the AFSOR grant ``Algorithmic and Logical Aspects when Composing Meanings'' and the FQXi grant ``Categorical Compositional Physics''.

\bibliographystyle{IEEEtran}
\bibliography{Relations}

\appendix
We outline proofs of the key results.
\relisacategory*
\begin{proof}
	The proof that~\relvq{(\Sigma,E)}{Q} is a category follows from the unit and associative properties of quantale multiplication. 
	We must first confirm that identities respect algebraic structure. For~$\sigma \in \Sigma$
	\begin{align*}
	1_A(a_1,a'_1) &\otimes ... \otimes 1_A(a_n,a'_n) = \\
	&= \left[ \bigvee \{ k \mid a_1 = a'_1 \} \right] \otimes ... \otimes \left[ \bigvee \{ k \mid a_n = a'_n \} \right] \\
	&= \bigvee \{ k \mid (a_1 = a'_1) \wedge ... \wedge (a_n = a'_n) \}\\
	&\leq \bigvee \{ k \mid \sigma(a_1,...,a_n) = \sigma(a'_1,...,a'_n) \}\\
	&= 1_A(\sigma(a_1,...,a_n), \sigma(a'_1,...,a'_n))
	\end{align*}
	We do the same for compositions
	\begin{align*}
	&(S \circ R)(a_1,c_1) \otimes ... \otimes (S \circ R)(a_n,c_n) =\\
	&= \left[\bigvee_{b_1}R(a_1,b_1) \otimes S(b_1,c_1)\right] \otimes \dots \\
	&\qquad \qquad \dots \otimes \left[\bigvee_{b_n}R(a_n,b_n) \otimes S(b_n,c_n)\right]\\
	&= \bigvee_{b_1,...,b_n} R(a_1,b_1) \otimes S(b_1,c_1) \otimes \dots \otimes R(a_n,b_n) \otimes S(b_n,c_n)\\
	&= \bigvee_{b_1,...,b_n} \left[ R(a_1,b_1) \otimes ... \otimes R(a_n,b_n) \right] \otimes \\
	&\qquad \qquad \qquad \otimes \left[ S(b_1,c_1) \otimes ... \otimes S(b_n,c_n)\right]\\
	&\leq \bigvee_{b_1,...,b_n} R(\sigma(a_1,...,a_n),\sigma(b_1,...,b_n)) \otimes\\
	&\qquad \qquad \qquad \otimes S(\sigma(b_1,...,b_n),\sigma(c_1,...,c_n))\\
	&\leq \bigvee_b R(\sigma(a_1,...,a_n),b) \otimes S(b,\sigma(c_1,...,c_n))\\
	&= (S \circ R)(\sigma(a_1,...,a_n), \sigma(c_1,...,c_n))
	\end{align*}
	Now we check explicitly the associativity of composition. For relations $R, S, T$, we have
	\begin{align*}
	(T \circ (S \circ R))(a,d) &= \bigvee_c (S \circ R)(a,c) \otimes T(c,d)\\
	&=\bigvee_c \left[ \bigvee_d R(a,b) \otimes S(b,c) \right] \otimes T(c,d)\\
	&=\bigvee_d R(a,b) \otimes \left[ \bigvee_c S(b,c) \otimes T(c,d) \right]\\
	&=\bigvee_d R(a,b) \otimes (T \circ S)(b,d)\\
	&=((T \circ S) \circ R) (a,d)
	\end{align*}
	We also check the right identity law.
	\begin{align*}
	(R \circ 1_A) (a,b) &= \bigvee_{a'} 1_A(a,a') \otimes R(a',b)\\
	&= \bigvee_{a'} \left[ \bigvee \{ k \mid a = a' \} \right] \otimes R(a',b)\\
	&= \bigvee_{a'} \bigvee \{ k \otimes R(a',b) \mid a = a' \}\\
	&= \bigvee_{a'} \{ R(a,b) \}\\
	&= R(a,b)
	\end{align*}
	The left identity law proof is similar. 
	
	Now we prove this category is symmetric monoidal. We define the monoidal unit to be the terminal algebra. 
	On objects, the tensor takes products of algebras. We define the action
	on morphisms pointwise as
	\begin{equation*}
	(R \otimes R')(a,a',b,b') = R(a,b) \otimes R'(a',b')
	\end{equation*}
	We first confirm~$R \otimes R'$ respects the algebraic structure.
	For~$\sigma \in \Sigma$ with arity~$n$
	\begin{align*}
	&(R \otimes R')(a_1,a'_1,b_1,b'_1) \otimes ... \otimes (R \otimes R')(a_n,a'_n,b_n,b'_n)=\\
	&= R(a_1,b_1) \otimes R'(a'_1,b'_1) \otimes ... \otimes R(a_n,b_n) \otimes R'(a'_n,b'_n)\\
	&= \left[ R(a_1,b_1) \otimes ... \otimes R(a_n,b_n) \right] \otimes \left[ R'(a'_1,b'_1) \otimes ... \otimes R'(a'_n,b'_n) \right]\\
	&\leq R(\sigma(a_1,...,a_n),\sigma(b_1,...,b_n)) \otimes\\
	&\qquad\qquad \otimes R'(\sigma(a'_1,...,a'_n),\sigma(b'_1,...,b'_n))\\
	&= (R \otimes R')((\sigma(a_1,...,a_n), \sigma(a'_1,...,a'_n)), \\
	&\qquad\qquad\qquad\qquad (\sigma(b_1,...,b_n),\sigma(b'_1,...,b'_n)))\\
	&= (R \otimes R')(\sigma((a_1,a'_1),...,(a_n,a'_n)), \sigma((b_1,b'_1),...,(b_n,b'_n)))
	\end{align*}
	Then, we show that the tensor is bifunctorial. Identities are preserved:
	\begin{align*}
	(1_{A_1} \otimes 1_{A_2})(a_1,a_2,&a'_1,a'_2) = 1_{A_1}(a_1,a'_1) \otimes 1_{A_2}(a_2,a'_2)\\
	&= \bigvee\{k \mid a_1 = a'_1\} \otimes \bigvee\{k \mid a_2 = a'_2\}\\
	&= \bigvee\{k \mid (a_1 = a'_1) \wedge (a_2 = a'_2)\}\\
	&= \bigvee\{k \mid (a_1, a_2) = (a'_1, a'_2)\}\\
	&= 1_{A_1 \otimes A_2}(a_1,a_2,a'_1,a'_2)
	\end{align*}
	For composition,
	\begin{align*}
	&\left[(S_1 \otimes S_2) \circ (R_1 \otimes R_2)\right](a_1,a_2,c_1,c_2) = \\
	&= \bigvee_{b_1,b_2} (R_1 \otimes R_2)(a_1,a_2,b_1,b_2) \otimes (S_1 \otimes S_2)(b_1,b_2,c_1,c_2)\\
	&= \bigvee_{b_1,b_2} R_1(a_1, b_1) \otimes R_2(a_2,b_2) \otimes S_1(b_1,c_1) \otimes S_2(b_2,c_2)\\
	&= \bigvee_{b_1,b_2} R_1(a_1,b_1) \otimes S_1(b_1,c_1) \otimes R_2(a_2,b_2) \otimes S_2(b_2,c_2)\\
	&= \left[ \bigvee_{b_1} R_1(a_1,b_1) \otimes S_1(b_1,c_1) \right] \otimes \\
	&\qquad \qquad \otimes \left[ \bigvee_{b_2} R_2(a_2,b_2) \otimes S_2(b_2,c_2) \right]\\
	&= (S_1 \circ R_1)(a_1,c_1) \otimes (S_2 \circ R_2)(a_2, c_2)\\
	&= \left[ (S_1 \circ R_1) \otimes (S_2 \circ R_2) \right](a_1,a_2,c_1,c_2)
	\end{align*}
	We consider~$\mathcal{E}$ as a symmetric monoidal category with respect to our choice
	of binary products and terminal object. We then take the graphs (see proposition~\ref{prop:relgraph} for the definition) 
	of the corresponding left and right unitors, associator and symmetry as the corresponding structure in~\relq{Q}.
	
	We must confirm that these coherence morphisms are natural in their parameters. The proofs are all
	similar, we check the associator explicitly
	\begin{align*}
	R \otimes &(S \otimes T) \circ \alpha_{A,B,C} = \\
	&= \bigvee_{x,y,z} \alpha_{A,B,C}(((a,b),c), (x,(y,z))) \otimes \\
	&\qquad \qquad \otimes \left[R \otimes (S \otimes T)\right]((x,(y,z)),(a',(b',c')))\\
	&= \bigvee_{x,y,z} \left[ \bigvee \{ k \mid (a = x) \wedge (b = y) \wedge (c =z) \} \right] \otimes\\
	& \qquad \qquad \qquad \qquad \otimes R(x,a') \otimes S(y,b') \otimes T(z,c')\\
	&= \bigvee_{x,y,z} \bigvee \{ R(x,a') \otimes S(y,b') \otimes T(z,c') \mid\\
	& \qquad \qquad \qquad \qquad \mid (a = x) \wedge (b = y) \wedge (c = z) \}\\
	&= R(a,a') \otimes S(b,b') \otimes T(c,c')\\
	&= \bigvee_{x,y,z} \bigvee \{ R(a,x) \otimes S(b,y) \otimes T(c,z) \mid\\
	& \qquad \qquad \qquad \qquad \mid (x = a') \wedge (y = b') \wedge (z = c') \}\\
	&= \bigvee_{x,y,z} R(a,x) \otimes S(b,y) \otimes T(c,z) \otimes \\
	& \qquad \qquad \otimes \left[ \bigvee \{ k \mid x = a' \wedge y = b' \wedge z = c' \} \right]\\
	&= \bigvee_{x,y,z} \left[ (R \otimes S) \otimes T \right](((a,b),c), ((x,y),z)) \otimes\\
	& \qquad \qquad \qquad \otimes \alpha_{A',B',C'}(((x,y),z), (a',(b',c')))\\
	&= \alpha_{A',B',C'} \circ (R \otimes S) \otimes T
	\end{align*}
	These are isomorphisms by functoriality of graphs (see proposition~\ref{prop:relgraph} for the proof). 
	Their inverses are given by their converses, as established in proposition~\ref{prop:relconverse}.
	
	Moreover, taking graphs commutes with our choice of products in~$\mathcal{E}$ in the sense that
	\begin{equation*}
	(f \times g)_\circ = f_\circ \otimes g_\circ
	\end{equation*}
	To check this we reason as follows:
	\begin{align*}
	(f \times g)_\circ &((a,a'),(b,b')) = 
	\bigvee \{ k \mid (b,b') = (f \times g)(a,a') \}\\
	&= \bigvee \{ k \mid (b = f(a)) \wedge (b' = g(a')) \}\\
	&= \left[ \bigvee \{ k \mid b = f(a) \} \right] \otimes \left[ \bigvee \{ k \mid b' = g(a') \} \right]\\
	&= f_\circ(a,b) \otimes g_\circ(a',b')
	\end{align*}
	This guarantees that the triangle and pentagon equations hold 
	as the same equations hold for the cartesian monoidal structure in~$\mathcal{E}$. The coherence conditions
	for the symmetry follow similarly.
\end{proof}

\relgraph*
\begin{proof}
	First of all we have to check that the resulting relation respects the algebraic structure. For
	$\sigma \in \Sigma$
	\begin{align*}
	f_\circ(a_1,b_1) \otimes &\dots \otimes f_\circ(a_n, b_n)=\\
	&= \left[ \bigvee \{ k \mid f(a_1) = b_1 \} \right] \otimes \dots \\
	&\qquad \dots \otimes \left[ \bigvee \{ k \mid f(a_n) = b_n \} \right]\\
	&= \bigvee \{ k \mid f(a_1) = b_1 \wedge \dots \wedge f(a_n) = b_n \}\\
	&\leq \bigvee \{ k \mid \sigma(f(a_1),\dots,f(a_n)) = \sigma(b_1,\dots,b_n) \}\\
	&= \bigvee \{ k \mid f(\sigma(a_1,\dots,a_n)) = \sigma(b_1,\dots,b_n)\\
	&= f_\circ(\sigma(a_1,\dots,a_n), \sigma(b_1,\dots,b_n))
	\end{align*}
	The we confirm this is functorial with respect to identities
	\begin{align*}
	\graph{1_A}(a_1,a_2) &= \bigvee \{ k \mid 1_A(a_1) = a_2 \}\\
	&= \bigvee \{ k \mid a_1 = a_2 \}\\
	&= 1_A(a_1,a_2)
	\end{align*}
	For functoriality with respect to composition,
	\begin{align*}
	(\graph{g} \circ \graph{f})&(a,c) =
	\bigvee_b f_\circ(a,b) \otimes g_\circ(b,c)\\
	&= \bigvee_b \left[ \bigvee \{ k \mid f(a) = b \} \right] \otimes \left[ \bigvee \{ k \mid g(b) = c \} \right]\\
	&= \bigvee \{ k \mid g(f(a)) = c \}\\
	&= \graph{(g \circ f)}(a,c)
	\end{align*}
	Finally we prove the preservation of the monoidal structure:
	\begin{align*}
	(f \times g)_\circ& ((a,a'),(b,b'))= \bigvee \{ k \mid (b,b') = (f \times g)(a,a') \}\\
	&= \bigvee \{ k \mid b = f(a) \wedge b' = g(a') \}\\
	&= \left[ \bigvee \{ k \mid b = f(a) \} \right] \otimes \left[ \bigvee \{ k \mid b' = g(a') \} \right]\\
	&= f_\circ(a,b) \otimes g_\circ(a',b')\\
	&= (f_\circ \otimes g_\circ)((a,a'),(b,b'))\qedhere
	\end{align*}
\end{proof}

\relconverse*
\begin{proof}
	We first check that taking the converse gives a well defined relation,.
	For~$\sigma \in \Sigma$ of arity~$n$, we reason
	\begin{align*}
	R^\circ(a_1,b_1) \otimes ... \otimes &R^\circ(a_n,b_n) =
	R(b_1,a_1) \otimes ... \otimes R(b_n,a_n)\\
	&\leq R(\sigma(b_1,...,b_n),\sigma(a_1,...,a_n))\\
	&= R^\circ(\sigma(a_1,...,a_n),\sigma(b_1,...,b_n))
	\end{align*}
	Next, we must confirm identities are preserved.
	\begin{align*}
	1_A^\circ (a_1,a_2) &=
	1_A(a_2,a_1)\\
	&=  \bigvee \{ k \mid a_2 = a_1 \}\\
	&=  1_A(a_1, a_2)
	\end{align*}
	We confirm also functoriality with respect to composition
	\begin{align*}
	(R \circ S)^\circ(a,c) &=
	(R \circ S)(c,a)\\
	&= \bigvee_b S(c,b) \otimes R(b,a)\\
	&= \bigvee_b S^\circ(b,c) \otimes R^\circ(a,b)\\
	&= \bigvee_b R^\circ(a,b) \otimes S^\circ(b,c)\\
	&= (S^\circ \circ R^\circ)(a,c)
	\end{align*}
	Finally, we must check that the converse distributes over tensors
	\begin{align*}
	(R^\circ \otimes S^\circ)(b,b',a,a') &=
	R^\circ(b,a) \otimes S^\circ(b',a')\\ 
	&= R(a,b) \otimes S(a',b')\\
	&= (R \otimes S)(a,a',b,b')\\
	&= (R \otimes S)^\circ(b,b',a,a')
	\end{align*}
	We moreover prove a fact used in the proof of proposition~\ref{prop:relisacategory}, that is,
	if~$f$ is an isomorphism in~$\mathcal{E}$ then
	\begin{equation*}
	(f^{-1})_\circ = (f_\circ)^\circ
	\end{equation*}
	This is a simple check:
	\begin{align*}
	(f^{-1})_\circ (b,a) &=\\
	&= \bigvee \{ k \mid f^{-1} b = a \}\\
	&= \bigvee \{ k \mid f(a) = b \}\\
	&= f_\circ(a,b)\\
	&= (f_\circ)^\circ (b,a)\qedhere
	\end{align*}
\end{proof}

%

\relvqhypergraph*
\begin{proof}
	For every object $A$ of $\mathcal{E}$, call $\epsilon_A$, $\delta_A$ the comultiplication and counit
	of proposition~\ref{prop:canonicalcomonoid}, and $\eta_A$, $\mu_A$
	their respective converses. The morphisms  $\epsilon_A$, $\delta_A$ have in the internal logic the explicit form
	\begin{gather*}
	\epsilon_A(a,x) = k \\
	\delta_A(a_1,(a_2,a_3)) = \bigvee \{ k \mid a_1 = a_2 = a_3 \}
	\end{gather*}
	Checking that $\eta_A$, $\mu_A$ form a monoid is straightforward, from the definition of converse.
	With respect to this monoid/comonoid pair, we first confirm the special axiom
	\begin{align*}
	&(\mu_A \circ \delta_A) (a_1,a_2) =\\
	&= \bigvee_{(a,a')} \delta_A(a_1,(a,a')) \otimes \mu_A((a,a'),a_2)\\
	&= \bigvee_{(a,a')} \left[ \bigvee \{ k \mid a_1 = a = a' \} \right] \otimes \left[ \bigvee \{ k \mid a = a' = a_2 \} \right]\\
	&= \bigvee_{(a,a')} \bigvee \{ k \mid a_1 = a =  a' = a_2 \}\\
	&= \bigvee \{ k \mid a_1 = a_2 \}\\
	&= 1_A(a_1, a_2)
	\end{align*}
	Finally, we check the Frobenius axiom, omitting some stages where the expressions get very long
	\begin{align*}
	&((1_A \otimes \delta_A) \circ (\mu_A \otimes 1_A)) (a_1, a_2,a_3,a_4) =\\
	&= \bigvee_{x,y,z} 1_A(a_1,x) \otimes \delta_A(a_2,(y,z)) \otimes \mu_A((x,y),a_3) \otimes 1_A(z,a_4)\\
	&= \bigvee_{x,y,z} \{ k \mid x = y = z = a_1 = a_2 = a_3 = a_4 \}\\
	&= \bigvee \{ k \mid a_1 = a_2 = a_3 = a_4 \}\\
	&= ((\delta_A \otimes 1_A) \circ (1_A \otimes \mu_A))(a_1, a_2,a_3,a_4)\qedhere
	\end{align*}
\end{proof}

\begin{lemma}
	\label{lem:qspaniso}
	Let~$\mathcal{E}$ be a finitely complete category, and~$(Q,\otimes,k)$ an internal monoid. If
	\begin{equation*}
	h : (X_1, f_1, g_1, \chi_1) \rightarrow (X_2, f_2, g_2, \chi_2)
	\end{equation*}
	is a~$Q$-span morphism with an inverse in~$\mathcal{E}$, then it is an isomorphism.
\end{lemma}
\begin{proof}
	We aim to show that~$h^{-1}$ is the required inverse as a~$Q$-span morphism. We calculate
	\begin{gather*}
	f_1 \circ h^{-1} =
	f_2 \circ h \circ h^{-1} =
	f_2\\
	g_1 \circ h^{-1} =
	g_2 \circ h \circ h^{-1} =
	g_2\\
	\chi_1 \circ h^{-1} =
	\chi_2 \circ h \circ h^{-1}=
	\chi_2\qedhere
	\end{gather*}
\end{proof}

\begin{lemma}
	\label{lem:monoidspancat}
	Let~$\mathcal{E}$ be a finitely complete category and~$(Q,\otimes,k)$ an internal monoid. $\spanq{Q}$ is a category.
\end{lemma}
\begin{proof}
	We first confirm that composition is independent of representatives. Consider span isomorphisms
	\begin{align*}
	\varphi : (X_1,f_1,g_1,\chi_1) &\rightarrow (X_2,f_2,g_2,\chi_2)\\
	\psi : (Y_1,h_1,k_1,\xi_1) &\rightarrow (Y_2,h_2,k_2,\xi_2)
	\end{align*}
	We must show an isomorphism between the $Q$-spans
	\begin{gather*}
	(X_1 \times_B Y_1,f_1 \circ p_1, k_1 \circ p_2, \otimes \circ (\chi_1 \times \xi_1) \circ \langle p_1, p_2 \rangle) \\
	(X_2 \times_B y_2,f_2 \circ p'_1, k_2 \circ p'_2, \otimes \circ (\chi_2 \times \xi_2) \circ \langle p'_1, p'_2 \rangle) 
	\end{gather*}
	We calculate, using properties of pullbacks
	\begin{equation*}
	f_2 \circ p'_1 \circ \varphi \times_B \psi =
	f_2 \circ \varphi \circ p_1 =
	f_1 \circ p_1
	\end{equation*}
	\begin{equation*}
	g_2 \circ p'_2 \circ \varphi \times_B \psi=
	g_2 \circ \psi \circ p_2=
	g_1 \circ p_2
	\end{equation*}
	Also
	\begin{align*}
	\otimes \circ (\chi_2 \times \xi_2) \circ &\langle p'_1, p'_2 \rangle \circ \varphi \times_B \psi =\\
	&= \otimes \circ (\chi_2 \times \xi_2) \circ (\varphi \times \psi) \circ \langle p_1, p_2 \rangle \\
	&= \otimes \circ (\chi_2 \circ \varphi, \xi_2 \circ \psi) \circ \langle p_1, p_2 \rangle \\
	&= \otimes \circ (\chi_1 \times \psi_1) \circ \langle p_1, p_2 \rangle
	\end{align*}
	and as~$\varphi \times_B \psi$ is also an isomorphism we can use lemma~\ref{lem:qspaniso} to complete this part of the proof.
	
	Next we confirm the left identity axiom. Firstly we note
	\begin{align*}
	(B,1_B,1_B&,\chi_k) \circ (X,f,g,\chi) =\\ 
	&= (X \times_B B, f \circ p_1, p_2, \otimes \circ (\chi \times \chi_k) \circ \langle p_1, p_2 \rangle)
	\end{align*}
	We claim~$p_1$ is a~$Q$-span morphism to the span~$(X,f,g,\chi)$. The conditions for this being a span morphism are
	\begin{equation*}
	f \circ p_1 = f \circ p_1 \quad\text{ and }\quad g \circ p_1 = p_2
	\end{equation*}
	and the second condition is obvious from the pullback square. Finally, we must confirm this also commutes with the characteristic functions.
	\begin{align*}
	\otimes \circ (\chi \times \chi_k) \circ \langle p_1, p_2 \rangle &= 
	p_1 \circ (\chi \times !) \circ \langle p_1, p_2 \rangle\\
	&= p_1 \circ \langle \chi \circ p_1, ! \circ p_2 \rangle\\
	&= \chi \circ p_1
	\end{align*}
	We must prove that~$p_1$ is an isomorphism, the required inverse being given by the universal property of pullbacks as
	\begin{equation*}
	\langle 1, g \rangle
	\end{equation*}
	Checking that this is an isomorphism follows from the universal property of pullbacks. We can then use lemma~\ref{lem:qspaniso} to complete
	this part of the proof. The right identity axiom follows similarly.
	
	We must then confirm associativity. We consider the composites
	\begin{align*}
	L &= ((Z,l,m,\zeta) \circ (Y,h,k,\xi)) \circ (X,f,g,\chi)\\
	R &= (Z,l,m,\zeta) \circ ((Y,h,k,\xi) \circ (X,f,g,\chi))
	\end{align*}
	Via the usual proof for categories of ordinary spans
	\begin{equation*}
	\iota := \langle p_1 \circ p_1, \langle p_2 \circ p_1, p_2 \rangle \rangle : (X \times_B Y) \times_C Z \rightarrow X \times_B (Y \times_C Z)
	\end{equation*}
	is an isomorphism of spans. It remains to show that this commutes with the characteristic morphisms. This is a horrible exercise in tracking
	various canonical morphisms, and is most easily handled using the graphical calculus for a cartesian monoidal category. Details are omitted to
	avoid a long typesetting exercise for the diagrams.
\end{proof}

\begin{lemma}
	\label{lem:spanqconverse}
	Let~$\mathcal{E}$ be a finitely complete category and~$(Q,\otimes,k)$ an internal monoid.
	There is an identity on objects contravariant involution (dagger functor) given by
	\begin{align*}
	\converse{(-)} : \spanq{Q}^{op} &\rightarrow \spanq{Q}\\
	(X,f,g,\chi) &\mapsto (X,g,f,\chi)
	\end{align*}
\end{lemma}
\begin{proof}
	This involution clearly preserves identities. We aim to show
	\begin{equation*}
	\converse{(X,f,g,\chi)} \circ \converse{(Y,h,k,\xi)} = \converse{((Y,h,k,\xi) \circ (X,f,g,\chi))}
	\end{equation*}
	These two spans are given by
	\begin{align*}
	\converse{(X,f,g,\chi)}& \circ \converse{(Y,h,k,\xi)} =\\
	&= (Y \times_B X, k \circ p_1, f \circ p_2, \otimes \circ \langle \xi \circ p_1, \chi \circ p_2 \rangle)\\
	((Y,h,k,\xi) \circ &(X,f,g,\chi))^\circ =\\
	&= (X \times_B Y, k \circ p_2, f \circ p_1, \otimes \circ \langle \chi \circ p_1, \xi \circ p_2 \rangle)
	\end{align*}
	The morphisms
	\begin{gather*}
	\langle p_2, p_1 \rangle : X \times_B Y \rightarrow Y \times_B X\\
	\langle p_2, p_1 \rangle : Y \times_B X \rightarrow X \times_B Y
	\end{gather*}
	witness an isomorphism in~$\mathcal{E}$. We first confirm that this gives a span isomorphism. This follows trivially from elementary properties of pullbacks.
	Finally, we must prove that this commutes with characteristic morphisms. This makes essential use of commutativity of $\otimes$
	\begin{align*}
	\otimes \circ \langle \xi \circ p_1, \chi \circ p_2 \rangle \circ \langle p_2, p_1 \rangle &=
	\otimes \circ \langle \xi \circ p_2, \chi \circ p_1 \rangle\\
	&= \otimes \circ \langle \chi \circ p_1, \xi \circ p_2 \rangle\qedhere
	\end{align*}
\end{proof}

\begin{lemma}
	\label{lem:monoidspangraph}
	Let~$\mathcal{E}$ be a finitely complete category, and~$(Q,\otimes,k)$ an internal monoid. There is an identity on objects covariant graph functor
	\begin{align*}
	(-)_\circ : \mathcal{E} &\rightarrow \spanq{Q}\\
	f : A \rightarrow B &\mapsto (A,1_A,f,\chi_k)
	\end{align*}
	There is also an identity on objects contravariant cograph functor
	\begin{align*}
	\cograph{(-)} : \mathcal{E}^{op} &\rightarrow \spanq{Q}\\
	f : A \rightarrow B &\mapsto (B,f,1_B,\chi_k)
	\end{align*}
	If~$Q$ is a commutative monoid then
	\begin{equation*}
	\cograph{(-)} = \converse{(-)} \circ \graph{(-)}
	\end{equation*}
\end{lemma}
\begin{proof}
	This construction is well known for ordinary spans. For~$Q$-spans, we must confirm the interaction with characteristic morphisms
	behaves appropriately. Firstly we note that
	\begin{equation*}
	(1_A)_\circ = (A,1_A,1_A,\chi_k)
	\end{equation*}
	and so identities are preserved.
	For composition, we have an ordinary span isomorphism
	\begin{equation*}
	\langle 1_A, f \rangle : A \rightarrow A \times_B B
	\end{equation*}
	We must confirm this commutes with characteristic morphisms
	\begin{align*}
	\otimes \circ (\chi_k \times \chi_k) \circ \langle p_1, p_2 \rangle \circ \langle 1_A, f \rangle
	&= \otimes \circ \langle \chi_k, \chi_k \circ f \rangle\\
	&= \otimes \circ \langle \chi_k, \chi_k \rangle\\
	&= \chi_k
	\end{align*}
	The proof for the cograph construction is similar.
	In the case of commutative~$Q$, the relationship to the converse functor is immediate from the definitions.
\end{proof}

\begin{lemma}
	\label{lem:spanqtensor}
	Let~$\mathcal{E}$ be a finitely complete category and~$(Q,\otimes,k)$ an internal commutative monoid. There is a bifunctor
	\begin{align*}
	\otimes : \spanq{Q} \times \spanq{Q} &\rightarrow \spanq{Q}\\
	A \otimes B &= A \times B\\
	(X_1,f_1,g_1,\chi_1) \otimes (X_2,&f_2,g_2,\chi_2) =\\
	=(X_1 \times X_2, f_1 &\times f_2, g_1 \times g_2, \otimes \circ (\chi_1 \times \chi_2))
	\end{align*}
	Furthermore, this bifunctor commutes with graphs in the sense that the following diagram commutes
	\begin{equation*}
	\begin{tikzpicture}[scale=0.5, node distance=2cm,->]
	\node (tl) {$\spanq{Q} \times \spanq{Q}$};
	\node[right of=tl, node distance=3cm] (tr) {$\spanq{Q}$};
	\node[below of=tl] (bl) {$\mathcal{E} \times \mathcal{E}$};
	\node[below of=tr] (br) {$\mathcal{E}$};
	\draw (tl) to node[above]{$\otimes$} (tr);
	\draw (bl) to node[below]{$\times$} (br);
	\draw (bl) to node[left]{$(-)_\circ \times (-)_\circ$} (tl);
	\draw (br) to node[right]{$(-)_\circ$} (tr);
	\end{tikzpicture}	
	\end{equation*}
\end{lemma}
\begin{proof}
	We first show that this respects equivalence classes of spans. Assume we have span isomorphisms
	\begin{gather*}
	\varphi : (X_1,f_1,g_1,\chi_1) \rightarrow (X'_1,f'_1,g'_1,\chi'_1) \\
	\psi : (X_2,f_2,g_2,\chi_2) \rightarrow (X'_2,f'_2,g'_2,\chi'_2)
	\end{gather*}
	The product~$\varphi \times \psi$ gives an isomorphism of ordinary spans
	\begin{align*}
	\varphi \times \psi : (X_1,f_1,g_1,\chi_1) &\otimes (X_2,f_2,g_2,\chi_2) \rightarrow \\
	&\rightarrow (X'_1,f'_1,g'_1,\chi'_1) \otimes (X'_2,f'_2,g'_2,\chi'_2)
	\end{align*}
	It then remains to check this commutes with characteristic morphisms. We calculate
	\begin{align*}
	\otimes \circ (\chi'_1 \times \chi'_2) \circ (\varphi \times \psi) &=
	\otimes \circ \left[ (\chi'_1 \circ \varphi) \times (\chi'_2 \circ \psi) \right]\\
	&= \otimes \circ (\chi_1 \times \chi_2)
	\end{align*}
	That this is bifunctorial as an operation on the underlying spans is well known. It remains to check the behaviour with respect to characteristic morphisms.
	For identity~$Q$-spans, the resulting characteristic function is
	\begin{equation*}
	\otimes \circ (\chi_k \times \chi_k) = \chi_k
	\end{equation*}
	For composition, we note there is an isomorphism of spans
	\begin{align*}
	&\langle \langle p_1 \circ p_1, p_1 \circ p_2 \rangle, \langle p_2 \circ p_1, p_2 \circ p_2 \rangle \rangle:\\
	&: (X \times_B Y) \times (X' \times_{B'} Y')
	\rightarrow (X \times X') \times_{B \times B'} (Y \times Y')
	\end{align*}
	We must show this commutes with the corresponding characteristic morphisms. The following unpleasant calculation establishes the required equality
	\begin{align*}
	\otimes \circ (\otimes \times \otimes) &\circ \left[ ((\chi \times \xi) \circ \langle p_1, p_2 \rangle) \times ((\chi' \times \xi') \langle p_1, p_2 \rangle) \right] =\\
	&= \otimes \circ (\otimes \times \otimes) \circ \left[ (\chi \times \xi) \times (\chi' \times \xi') \right] \circ \\
	&\qquad \circ \langle \langle p_1 p_1, p_2 p_1 \rangle, \langle p_1 p_2, p_2 p_2 \rangle \rangle\\
	&=\otimes \circ (\otimes \times \otimes) \circ \left[ (\chi \times \chi') \times (\xi \times \xi') \right] \circ\\
	&\qquad \circ \langle \langle p_1 p_1, p_1 p_2 \rangle, \langle p_2 p_1, p_2 p_2 \rangle \rangle \circ\\
	&\qquad \qquad \circ \langle \langle p_1 p_1, p_2 p_1 \rangle, \langle p_1 p_2, p_2 p_2 \rangle \rangle\\
	&=\otimes \circ (\otimes \times \otimes) \circ \left[ (\chi \times \chi') \times (\xi \times \xi') \right] \circ\\
	&\qquad \circ \langle \langle p_1 p_1, p_1 p_2 \rangle, \langle p_2 p_1, p_2 p_2 \rangle \rangle
	\end{align*}
	Finally, we must confirm that this bifunctor commutes with graphs. On objects this is obvious, as all the functors involved act as the identity on objects. On morphisms we reason
	\begin{align*}
	(f_1 \times f_2)_\circ &=
	(A_1 \times A_2, 1_{A_1 \times A_2}, f_1 \times f_2, k)\\
	&= (A_1, 1_{A_1}, f, k) \otimes (A_2, 1_{A_2}, f, k)\\
	&= (f_1)_\circ \otimes (f_2)_\circ\qedhere
	\end{align*}
\end{proof}

\begin{lemma}
	\label{lem:spancalc}
	Let~$\mathcal{E}$ be a finitely complete category and~$(Q,\otimes,k)$ an internal monoid. The following equation holds in~\spanq{Q}
	\begin{equation*}
	\graph{k} \circ (X,f,g,\chi) \circ \cograph{h} = (X, h \circ f, k \circ g, \chi)
	\end{equation*}
\end{lemma}
\begin{proof}
	We firstly consider the case of post composition with the graph of a morphism in the underlying category
	\begin{equation*}
	\graph{k} \circ (X,f,g,\chi)
	\end{equation*}
	This composite is given by the pullback span
	\begin{equation*}
	(X \times_B B, p_1 \circ f, p_2 \circ k, \otimes \circ (\chi \times \chi_k) \circ \langle p_1, p_2 \rangle)
	\end{equation*}
	We note that $p_1 \circ \langle 1_X, g \rangle = 1_X$
	and
	\begin{equation*}
	\langle 1_X, g \rangle \circ p_1 =
	\langle p_1, g \circ p_1 \rangle =
	\langle p_1, p_2 \rangle =
	1_{X \times_B B}
	\end{equation*}
	and so~$p_1$ and $\langle 1_X, g \rangle$ witness an isomorphism in~$\mathcal{E}$. We next confirm they give a span isomorphism. One
	of the conditions for~$p_1$ to be a span morphism is trivial, for the other
	\begin{equation*}
	k \circ g \circ p_1 =
	k \circ 1 \circ p_2 =
	k \circ p_2
	\end{equation*}
	Finally, we must confirm that this commutes with the characteristic morphisms
	\begin{equation*}
	\otimes \circ (\chi \times \chi_k) \circ \langle p_1, p_2 \rangle=
	\chi \circ p_1 \circ \langle p_1, p_2 \rangle=
	\chi \circ p_1
	\end{equation*}
	Now we note that
	\begin{align*}
	(X,f,g,\chi) \circ \cograph{h} &= (X,f,g,\chi) \circ h_\circ^\circ\\
	&= (h_\circ \circ (X,f,g,\chi)^\circ)^\circ\\
	&=(h_\circ \circ (X,g,f,\chi))^\circ\\
	&=(X,g,h \circ f, \chi)^\circ\\
	&=(X,h \circ f, g,\chi)
	\end{align*}
	Combining these two observations then completes the proof.
\end{proof}

\simplespans*
\begin{proof}
	We first confirm that~\spanq{Q} carries a monoidal structure. We take as our monoidal unit the terminal object in~$\mathcal{E}$ and the tensor
	to be the bifunctor of proposition~\ref{lem:spanqtensor}. Next, we note that the underlying category is a symmetric monoidal category with
	respect to binary products. The graph construction is surjective on objects, and commutes with the tensor, therefore the graphs of the coherence
	morphisms in~$\mathcal{E}$ lift to~\spanq{Q}. We must confirm that each of these remains natural in~\spanq{Q}.
	
	Applying proposition~\ref{lem:spancalc}, it is sufficient to show the following are span isomorphisms
	\begin{align*}
	&\lambda_X : (1 \times X, \lambda_A \circ (1_1 \times f), \lambda_B \circ (1_1 \times g), \\
	&\qquad \qquad \qquad \qquad ,\otimes \circ (\chi_k \times \chi)) \rightarrow (X,f,g,\chi)\\
	&\rho_X : (X \times 1, \rho_A \circ (f \times 1_1), \rho_B \circ (g \times 1_1), \\
	&\qquad \qquad \qquad \qquad ,\otimes \circ (\chi \times \chi_k)) \rightarrow (X,f,g,\chi)\\
	&\alpha_{X,Y,Z} : ((X_1 \times X_2) \times X_3, \alpha_{A_1,A_2,A_3} \circ\\
	&\qquad\circ ((f_1 \times f_2) \times f_3), \alpha_{B_1,B_2,B_3} \circ ((g_1 \times g_2) \times g_3), ...) \rightarrow\\
	&\qquad \qquad \qquad \rightarrow (X_1 \times (X_2 \times X_3))\\
	&\sigma_{X,Y} : (X_1 \times X_2, \sigma_{A_1, A_2} \circ f_1 \times f_2, \sigma_{B_1,B_2} \circ g_1 \times g_2,\\
	&\qquad,\otimes  \circ (\chi_1 \times \chi_2)) \rightarrow\\
	&\qquad\qquad\rightarrow (X_2 \times X_1, f_2 \times f_2, g_2 \times g_1, \otimes \circ (\chi_2 \times \chi_1)) 
	\end{align*}
	And this is now just a straightforward (but very unpleasant) check.
\end{proof}

\begin{lemma}
	Let~$\mathcal{E}$ be a topos, $(Q, \otimes, k, \leq)$ an internal partially ordered commutative monoid, and~$(\Sigma,E)$ an algebraic variety.
	If~$(X_1,f_1,g_1,\chi_1)$ is an algebraic~$Q$-span, and~$(X_2,f_2,g_2,\chi_2)$ is an isomorphic~$Q$-span, then
	it is also an algebraic span.
\end{lemma}
\begin{proof}
	For the assumptions in the question, with~$\iota$ denoting the assumed isomorphism and $\sigma \in \Sigma$, if
	\begin{equation*}
	f_2(x_1) = a_1 \wedge g_2(x_1) = b_1 \wedge ... \wedge f_2(x_n) = a_n \wedge g_2(x_n) = b_n
	\end{equation*}
	then using our span isomorphism.
	\begin{align*}
	f_1(\iota^{-1}(x_1)) = a_1 &\wedge g_1(\iota^{-1}(x_1)) = b_1 \wedge \dots\\
	&\dots \wedge f_1(\iota^{-1}(x_n)) = a_n \wedge g_1(\iota^{-1}(x_n)) = b_n
	\end{align*}
	By assumption that the first span is algebraic, there exists~$x$ such that
	\begin{align*}
	f_1(x) = &\sigma(a_1,...,a_n) \wedge g_1(x) = \sigma(b_1,...,b_n) \wedge \\
	&\wedge \chi_1(\iota^{-1}(x_1)) \otimes ... \otimes \chi_1(\iota^{-1}(x_n)) \leq \chi_1(x)
	\end{align*}
	Therefore, using our span isomorphism again
	\begin{align*}
	f_2(\iota(x)) &= \sigma(a_1,...,a_n) \wedge g_2(\iota(x)) = \sigma(b_1,...,b_n) \wedge \\
	&\wedge \chi_2(x_1) \otimes ... \otimes \chi_2(x_n) \leq \chi_2(\iota(x))\qedhere
	\end{align*}
\end{proof}

\begin{lemma}
	\label{lem:algebraicspancat}
	Let~$\mathcal{E}$ be a topos, $(Q, \otimes, k, \leq)$ an internal partially ordered commutative monoid, 
	and~$(\Sigma,E)$ an algebraic variety. \spanvq{(\Sigma,E)}{Q} is a category.
\end{lemma}
\begin{proof}
	Throughout, we will perform checks for an arbitrary~$\sigma \in \Sigma$.
	Firstly, we must confirm that the identity morphisms are algebraic. The required condition
	is immediate, as~$A$ is closed under the algebraic operations, and the characteristic morphism
	is constant in the internal language.
	
	Secondly, we must confirm that algebraic~$Q$-spans are closed under composition. Assume
	\begin{align*}
	f(p_1(z_1)) = a_1 &\wedge k(p_2(z_1)) = c_1 \wedge ...\\
	&\dots \wedge f(p_1(z_n)) = a_n \wedge k(p_2(z_n)) = c_n
	\end{align*}
	then there exist~$x_1,...,x_n$ and~$y_1,...,y_n$ such that
	\begin{align*}
	f(x_1) = a_1 &\wedge g(x_1) = h(y_1) \wedge k(y_1) = c_1 \wedge ...\\
	&\dots \wedge f(x_n) = a_n \wedge g(x_n) = h(y_n) \wedge k(y_n) = c_n
	\end{align*}
	therefore as the component spans are algebraic, there exist~$x$ and~$y$ such that
	\begin{equation*}
	f(x) = \sigma(a_1,...,a_n) \wedge g(x) = h(y) \wedge k(y) = \sigma(c_1,...,c_n)
	\end{equation*}
	and both
	\begin{equation*}
	\chi(x_1) \otimes ... \otimes \chi(x_n) \leq \chi(x), \qquad 	\xi(y_1) \otimes ... \otimes \xi(y_n) \leq \xi(y)
	\end{equation*}
	Therefore we have~$(x,y)$ in the apex of the composite span, with truth value $	\chi(x) \otimes \xi(y)$.
	By monotonicity of the tensor
	\begin{equation*}
	\chi(x_1) \otimes ... \otimes \chi(x_n) \otimes \xi(y_1) \otimes ... \otimes \xi(y_n) \leq \chi(x) \otimes \xi(y)
	\end{equation*}
	Finally, as the tensor is commutative
	\begin{equation*}
	\chi(x_1) \otimes \xi(y_1) \otimes ... \otimes \chi(x_n) \otimes \xi(y_n) \leq \chi(x) \otimes \xi(y)
	\end{equation*}
	That the composition is associative and the identities satisfy the required axioms follows immediately from the same properties
	for the underlying~$Q$-spans as established in lemma~\ref{lem:monoidspancat}.
\end{proof}

\spanvqsmc*
\begin{proof}
	Throughout, we will perform checks for an arbitrary~$\sigma \in \Sigma$. The proof is exactly as in lemma~\ref{lem:spanqtensor}, we just need to prove that the functors defined in lemmas~\ref{lem:spanqconverse}, \ref{lem:monoidspangraph}, \ref{lem:spanqtensor} respect the algebraic condition.
	
	With regard to lemma~\ref{lem:spanqconverse}, we just observe that the condition for a~$Q$-span to be algebraic is symmetrical in its
	domain and codomain, and therefore preserved.
	
	With regard to lemma~\ref{lem:monoidspangraph}, as the characteristic morphism is constant, we must simply confirm the existence of witnesses relating composite terms. If
	\begin{equation*}
	f(a_1) = b_1 \wedge ... \wedge f(a_n) = b_n
	\end{equation*}
	then as~$f$ is a homomorphism
	\begin{equation*}
	f(\sigma(a_1,...,a_n)) = \sigma(f(a_1),...,f(a_n)) = \sigma(b_1,...,b_n)
	\end{equation*}
	With regard to lemma~\ref{lem:spanqtensor}, for algebraic~$Q$-spans~$(X,f,g,\chi)$ and~$(X',f',g',\chi')$, if
	\begin{equation*}
	(f \times f')(x_1,x'_1) = (a_1,a'_1) \wedge ... \wedge (f \times f')(x_n,x'_n) = (a_n,a'_n)
	\end{equation*}
	and
	\begin{equation*}
	(g \times g')(x_1,x'_1) = (b_1,b'_1) \wedge ... \wedge (g \times g')(x_n,x'_n) = (b_n,b'_n)
	\end{equation*}
	Then
	\begin{equation*}
	f(x_1) = a_1 \wedge f'(x'_1) = a'_1 \wedge ... \wedge f(x_n) = a_n \wedge f'(x'_n) = a'_n
	\end{equation*}
	and
	\begin{equation*}
	g(x_1) = b_1 \wedge g'(x'_n) = b'_1 \wedge ... \wedge g(x_n) = b_n \wedge g'(x'_n) = b'_n
	\end{equation*}
	As the spans are algebraic, there exist~$x$ and~$x'$ such that
	\begin{equation*}
	\chi(x_1) \otimes ... \otimes \chi(x_n) \leq \chi(x), \qquad \chi'(x'_1) \otimes ... \otimes \chi'(x'_n) \leq \chi'(x')
	\end{equation*}
	By monotonicity and commutativity of the tensor, we then have
	\begin{equation*}
	\chi(x_1) \otimes \chi'(x'_1) \otimes ... \otimes \chi(x_n) \otimes \chi'(x'_n) \leq \chi(x) \otimes \chi'(x')
	\end{equation*}
	As we said, functoriality, and that the required diagram commutes then follows from the proof of theorem~\ref{thm:simplespans} as tensors coincide
	on the underlying~$Q$-spans.
\end{proof}

\conversealgebraic*
\begin{proof}
	Converse is defined as in lemma~\ref{lem:spanqconverse}, and the proof that this is algebraic is in proposition~\ref{prop:spanvqsmc}. That converse commutes with the tensor is trivial.
\end{proof}

\graphalgebraic*
\begin{proof}
	Graph is defined as in lemma~\ref{lem:monoidspangraph}, and the proof that this is algebraic is in proposition~\ref{prop:spanvqsmc}. The proof that graphs commute with the tensor is as in lemma~\ref{lem:spanqtensor}.
\end{proof}

\spanvqhypergraph*
\begin{proof}
	As in theorem~\ref{thm:simplespans}, since all the tools used there preserve the algebraic structure.
\end{proof}

\relcanbeordered*
\begin{proof}
	First of all we have to prove that our partial order is well defined.
	It is clearly reflexive. For transitivity, if
	\begin{equation*}
	R \subseteq R' \quad\text{ and }\quad R' \subseteq R''
	\end{equation*}
	then both
	\begin{equation*}
	\vdash R(a,b) \leq R'(a,b) \quad\text{ and }\quad \vdash R'(a,b) \leq R''(a,b)
	\end{equation*}
	Therefore internally
	\begin{equation*}
	\vdash R(a,b) \leq R'(a,b) \wedge R'(a,b) \leq R''(a,b)
	\end{equation*}
	and so by transitivity of the order on the quantale
	\begin{equation*}
	\vdash R(a,b) \leq R''(a,b)
	\end{equation*}
	Finally, if	$R \subseteq R'$ and $R' \subseteq R$,
	and so internally
	\begin{equation*}
	\vdash R(a,b) \leq R'(a,b) \wedge R'(a,b) \leq R(a,b)
	\end{equation*}
	by antisymmetry of the order on the internal quantale
	\begin{equation*}
	\vdash R(a,b) = R'(a,b)
	\end{equation*}
	and so
	\begin{equation*}
	\vdash \forall a,b. R(a,b) = R'(a,b)
	\end{equation*}
	Meaning externally $R = R'$.
	
	Next, we must confirm that composition is monotone in both components. As the proofs are symmetrical, we only consider precomposition explicitly.

	Assume $R \subseteq R'$.	
	We consider post-composing each of these relations with the relation $S$.
	Remembering that the quantale product preserves order, we calculate
	\begin{align*}
	(S \circ R)(a,c) &= \bigvee_b R(a,b) \otimes S(b,c)\\
	&\leq R'(a,b) \otimes S(b,c)\\
	&= \bigvee_b R'(a,b) \otimes S(b,c)\\
	&= (S \circ R')(a,c)
	\end{align*}
	Finally, we must confirm that the tensor on~\relvq{(\Sigma,E)}{Q} is monotone in both arguments.
	Assume again $R \subseteq R'$. We calculate
	\begin{align*}
	(R \otimes S)(a,b,c,d) &= R(a,b) \otimes S(c,d)\\
	&\leq R'(a,b) \otimes S(c,d)\\
	&= (R' \otimes S)(a,b,c,d)\qedhere
	\end{align*}
\end{proof}

\spancanbeordered*
\begin{proof}
	Firstly, we must confirm that this ordering is independent of choices of representatives for equivalence classes of spans.

	Assume $(X_1,f_1,g_1,\chi_1) \subseteq (Y,h_1,k_1,\xi_1)$,
	and span isomorphisms
	\begin{align*}
	i : (X_1, f_1, g_1, \chi_1) &\rightarrow (X_2, f_2, g_2, \chi_2)\\
	j : (Y_1, h_1, k_1, \xi_1) &\rightarrow (Y_2, h_2, k_2, \xi_2)
	\end{align*}
	Let
	\begin{equation*}
	m : (X_1, f_1, g_1, \chi_1) \rightarrow (Y,h_1,k_1,\xi_1)
	\end{equation*}
	be the span morphism that is monic in~$\mathcal{E}$ required by the assumed order structure.
	There is then a span morphism
	\begin{equation*}
	j \circ m \circ i^{-1} : (X_2, f_2, g_2, \chi_2) \rightarrow (Y_2, h_2, k_2, \xi_2)
	\end{equation*}
	which is monic in~$\mathcal{E}$ as monomorphisms are closed under composition. We then have
	\begin{equation*}
	\chi_2(x) = 
	\chi_1(i^{-1}(x)) =
	\xi_1(m \circ i^{-1}(x)) =
	\xi_2(j \circ m \circ i^{-1}(x))
	\end{equation*}
	The relation~$\subseteq$ is clearly reflexive via the identity~$Q$-span morphism. 
	
	For transitivity, assume
	\begin{equation*}
	(X,f,g,\chi) \subseteq (Y,h,k,\xi) \subseteq (Z,m,n,\zeta)
	\end{equation*}
	Denote the required monomorphisms
	\begin{align*}
	r : (X,f,g,\chi) &\rightarrow (Y,h,k,\xi)\\
	s : (Y,h,k,\xi) &\rightarrow (Z,m,n,\zeta)
	\end{align*}
	There is then an obvious span morphism $s \circ r$ that is a monomorphism in~$\mathcal{E}$.
	We then have
	\begin{equation*}
	\chi(x) \leq \xi(r(x))
	\end{equation*}
	and so
	\begin{equation*}
	\xi(r(x)) \leq \zeta(s \circ r(x))
	\end{equation*}
	By transitivity of the quantale ordering
	\begin{equation*}
	\chi(x) \leq \zeta(s \circ r (x))
	\end{equation*}
	Then, we must confirm that composition is monotone in both components. As the proofs are symmetrical, we only consider precomposition explicitly.
	
	Assume
	$(X_1,f_1,g_1,\chi_1) \subseteq (X_2,f_2,g_2,\chi_2)$
	as witnessed by~$\mathcal{E}$ monomorphism~$m : X_1 \rightarrow X_2$.
	We consider post-composing each of these spans with the span
	$(Y,h,k,\xi)$.
	There is then a~$Q$-span morphism
	\begin{align*}
	m \times_B 1 : \{ (x_1,y) \mid f_1(x_1) &= h(y) \} \rightarrow\\
	&\rightarrow \{ (x_2,y) \mid f_2(x_2) = h(y) \}
	\end{align*}
	and the underlying morphism is a monomorphism in~$\mathcal{E}$ by standard properties of pullbacks and monomorphisms.
	By monotonicity of the tensor
	\begin{align*}
	(\otimes \circ (\chi_1 \times \xi) &\circ \langle p_1, p_2 \rangle) (x,y)
	= \xi_1(x) \otimes \xi(y)\\
	&\leq \xi_2(m(x)) \otimes \xi(y)\\
	&= (\otimes \circ (\chi_2 \times \xi) \circ \langle p_1, p_2 \rangle) (m(x), y)\\
	&= (\otimes \circ (\chi_2 \times \xi) \circ \langle p_1, p_2 \rangle \circ (m \times_B 1)) (x,y)
	\end{align*}
	Finally, we must confirm that the tensor bifunctor on~\spanvq{(\Sigma,E)}{Q} is monotone in both arguments.
	Assume
	\begin{equation*}
	(X_1,f_1,g_1,\chi_1) \subseteq (X_2,f_2,g_2,\chi_2)
	\end{equation*}
	witnessed by~$\mathcal{E}$ monomorphism~$m : X_1 \rightarrow X_2$. There is then a span morphism
	\begin{align*}
	m \times 1 : (X_1,f_1,g_1,\chi_1) &\otimes (Y,h,k,\xi) \rightarrow\\
	&\rightarrow (X_2,f_2,g_2,\chi_2) \otimes (Y,h,k,\xi)
	\end{align*}
	and this a~$\mathcal{E}$ monomorphism by standard theory of products and monomorphisms. We then calculate
	\begin{align*}
	\otimes \circ (\chi_1 \times \xi) (x,y) &=
	\chi_1(x) \otimes \xi(y)\\
	&\leq \chi_2(m(x)) \otimes \xi(y)\\
	& = (\otimes \circ (\chi_2 \times \xi) \circ (m \times 1)) (x,y)\qedhere
	\end{align*}
\end{proof}

\vfunctor*
\begin{proof}
	We define the action on morphisms on a chosen representative span as follows
	\begin{equation*}
	V(X,p_1,p_2,\chi)(a,b) = \bigvee_x \{ \chi(x) \mid p_1(x) = a \wedge p_2(x) = b \}
	\end{equation*}
	It is easy to check that this definition is independent of our choice of representatives. 
	We moreover check preservation of identities:
	\begin{align*}
		V(A,1_A,1_A,\chi_k)(a_1, a_2) &= \bigvee \{ \chi_k(a) \mid 1_A(a) = a_1 \wedge 1_A(a) = a_2 \}\\
		&= \bigvee \{ k \mid a = a_1 \wedge a = a_2 \}\\
		&= \bigvee \{ k \mid a_1 = a_2 \}\\
		&= 1_A(a_1, a_2)
	\end{align*}
	That the functor commutes with the tensor is clear from the definition.
	That the coherence morphisms for the monoidal structures are preserved on the
	nose is clear as they were constructed using the graph constructions, and $V$
	preserves graphs. To see that $V$ preserves preorders, we assume
	\begin{equation*}
	(X_1,f_1,g_1,\chi_1) \subseteq (X_2,f_2,g_2,\chi_2)
	\end{equation*}
	witnessed by a monomorphism~$m : X_1 \rightarrow X_2$.
	We then calculate
	\begin{align*}
		\bigvee \{ \chi_1(x) &\mid f_1(x) = a \wedge g_1(x) = b \} \leq \\
		&\leq \bigvee \{ \chi_2(m(x)) \mid f_2(m(x)) = a \wedge g_2(m(x)) = b \}\\
		&\leq \bigvee \{ \chi_2(x') \mid f_2(x') = a \wedge g_2(x') = b \}\qedhere
	\end{align*}
\end{proof}

\quantalemorphisminduced*
\begin{proof}
	$h^*$ acts by postcomposing a relation $R: A \times B \to Q_1$ with $h$:
	\begin{equation*}
		h^*(R): A \times B \xrightarrow{R} Q_1 \xrightarrow{h} Q_2
	\end{equation*}
	Preservation of identities, compositions and tensors follows from the fact that $h$ preserves quantale identities,
	 products and joins. In particular, the preservation of quantale joins makes $h$ also order-preserving, 
	 proving that $h^*$ is a \pos-functor. 
	 The functoriality of assigments is trivial: Postcomposing with the identity function on a quantale 
	 gives back the same category we started from, and the composition of quantale homomorphisms is a quantale homomorphism.
\end{proof}

\preorderedmonoidmorphisminduced*
\begin{proof}
	The $h^*$ functor is again defined by postcomposition, as in the relational case. 
	Proof of functoriality and preservation of tensor and preorder is almost identical to the relational case.
\end{proof}

\relclosedunderterms*
\begin{proof}
	Suppose $R$ is a linear relation. we proceed by induction. 
	By definition, if $\tau$ is any $n$-ary operation, then it is
	\[
	\bigotimes_{k=1}^n R(a_k, b_k)  \leq R(\tau (a_1, \dots, a_n), \tau (b_1, \dots, b_n))
	\]
	(in this proof the big tensor symbol is just a shorthand for the quantale product over a finite number of components).
	Now let $\tau_{1}, \dots, \tau_{n}$ be operations of arities 
	$n, k_1, \dots, k_n$, respectively. Being $\tau$ an operation it is
	\begin{align*}
	\bigotimes_{i=1}^n R(\tau_i(a^i_1, &\dots a^i_{k_i}) , \tau_i(b^i_1, \dots b^i_{k_i})) \leq  \\
	&\leq R(\tau(\tau_{1}(a^1_1, \dots a^1_{k_1}), \dots, \tau_{n}(a^n_1, \dots a^i_{k_n})),\\
	&\qquad \qquad \tau (\tau_{1}(b^1_1, \dots b^1_{k_1}), \dots, \tau_{n}(b^n_1, \dots b^n_{k_n})))
	\end{align*}		
	and combining with the same condition on the $\tau_i$ one obtains
	\begin{align*}
	\bigotimes_{i=1}^n \bigotimes_{z=1}^{k_i} &R(a^i_{z}, b^i_{z}) \leq \\
	&\leq R(\tau(\tau_{1}(a^1_1, \dots a^1_{k_1}), \dots, \tau_{n}(a^n_1, \dots a^i_{k_n})),\\
	&\qquad \qquad \tau (\tau_{1}(b^1_1, \dots b^1_{k_1}), \dots, \tau_{n}(b^n_1, \dots b^n_{k_n}))
	\end{align*}
	This concludes the proof since every linear term can be written as a concatenation of operations. 
	
	Affine terms are obtained as compositions of operations and projections. 
	It is thus sufficient to prove that the condition holds for affine relations if 
	$\tau$ is a projection. Then, we can treat any n-ary projection as a generic operation 
	and proceed as in the previous case. But the condition
	\begin{equation*}
		\bigotimes_{k=1}^n R(a_k, b_k)  \leq R(\pi (a_1, \dots, a_n), \pi (b_1, \dots, b_n))
	\end{equation*} 
	being the right hand side just $R(a_i, b_i)$, trivially holds when $R$ is affine. 
	
	Relevant terms are obtained as compositions of operations and diagonals. 
	Note that a diagonal $\delta$ is not a term when taken alone because it is not a morphism of the form 
	$X^n \to X$ for some $n$. This means that if a term is built using diagonals, 
	there is always at least one operation that is composed with the diagonal on the left. 
	The proof is then very similar to the previous ones, 
	with the additional step that if $R$ is relevant, then the condition has to be proven to hold for every 
	$\tau(x_1, \dots, x_i, \delta(x), x_{i+1}, \dots, x_n)$, 
	where $\delta$ is the n-th diagonal and $\tau$ is any $(n+m)$-ary operation.
	
	For cartesian terms it is sufficient to put all these observation together and proceed in the same way.
\end{proof}

\relniceness*
\begin{proof}
	We only prove the affine case explicitly, all the rest being similar. 
	For arbitrary $a_1,a_2,b_1,b_2$ consider the product $R(a_1,b_1) \otimes R(a_2,b_2)$. 
	Both $R(a_1,b_1)$ and $R(a_2,b_2)$ are elements of $Q$, hence, 
	being $Q$ affine, we can readily infer $R(a_1,b_1) \otimes R(a_2,b_2) \leq R(a_1,b_1)$. 
	Being the variables arbitrarily chosen, we universally quantify on them obtaining 
	the axiom of being affine for $R$ as a valid formula in the ambient topos $\mathcal{E}$. 
\end{proof}

\interpretationinducedrel*
\begin{proof}
	An object of $\linrelvq{(\Sigma_2,E_2)}{Q}$ is written as
	$\langle\langle A,\sigma_j\rangle\rangle$, where $A$ is an object of $\mathcal{E}$ and the $\sigma_j$ are morphisms $A^n \to A$ in bijective correspondence
	with the operations in $\Sigma_2$, agreeing with them on arities, and such that they satisfy the equations in $E_2$ 
	(these equations are just commutative diagams between the above mentioned morphisms). 
	The linear (affine, relevant, cartesian) interpretation $i$ maps every operation in $\sigma'_k \in \Sigma_1$ 
	to a linear (affine, relevant cartesian) term $i(\sigma'_k)$ on $\Sigma_2$, 
	such that these terms satisfy the equations in $E_1$.
	This means that $\langle\langle A, i(\sigma'_k)\rangle\rangle$ is an algebra of type $(\Sigma_1, E_1)$. 
	
	The functor $i^*$ then acts as follows: It sends every  $\langle\langle A,\sigma_j\rangle\rangle$ to $\langle\langle A, i(\sigma'_k)\rangle\rangle$, and it is identity on morphisms 
	(the fact that morphisms of $\linrelvq{(\Sigma_2,E_2)}{Q}$ 
	are also morphisms of $\linrelvq{(\Sigma_1,E_1)}{Q}$ is a direct consequence 
	of proposition~\ref{prop:relclosedunderterms}).
	Functoriality then holds trivially being $i^*$ identity on morphisms.
\end{proof}

\spanclosedterms*
\begin{proof}
	As in the relational case, we proceed by induction. Let $(X,f,g,\chi)$ be a span.
	By definition, if $\tau$ is any $n$-ary operation, then the axiom
	(here the big wedge and the big tensor product are just a shorthand for a logical conjunction
	 and a quantale product over a finite number of components, respectively)
	\begin{gather*}
	\bigwedge_{i=1}^n (f(x_i) = a_i \wedge g(x_i) = b_i) \implies \\
	\exists x: f(x) = \tau(a_1, \dots, a_n), k(x) \wedge g(x) = \tau(b_1, \dots, b_n) \wedge\\
	\wedge \bigotimes_{i=1}^n \chi(x_i) \leq \chi(x)
	\end{gather*}
	is already satisfied.
	Now let $\tau, \tau_{1}, \dots, \tau_{n}$ be operations of arities 
	$n, k_1, \dots, k_n$, respectively. Being $\tau$ an operation it is
	\begin{gather*}
	\bigwedge_{i=1}^n (f(x^i) = \tau_i(a^i_1, \dots a^i_{k_i}) \wedge g(x^i) = \tau_i(b^i_1, \dots b^i_{k_i})) \implies \\
	\exists x: f(x) = \tau(\tau_{1}(a^1_1, \dots a^1_{k_1}), \dots, \tau_{n}(a^n_1, \dots a^i_{k_n})) \wedge \\ 
	\wedge g(x) = \tau(\tau_{1}(b^1_1, \dots b^1_{k_1}), \dots, \tau_{n}(b^n_1, \dots b^i_{k_n})) \wedge\\
	\wedge \bigotimes_{i=1}^n \chi(x_i) \leq \chi(x)
	\end{gather*}
	and combining with the same condition on the $\tau_i$ one obtains
	\begin{gather*}
	\bigwedge_{i=1}^n\bigwedge_{j=1}^{k_i} (f(x^i_j) = a^i_j \wedge g(x^i_j) = b^i_j) \implies \\
	\exists x: f(x) = \tau(\tau_{1}(a^1_1, \dots a^1_{k_1}), \dots, \tau_{n}(a^n_1, \dots a^i_{k_n})) \wedge \\ 
	\wedge g(x) = \tau(\tau_{1}(b^1_1, \dots b^1_{k_1}), \dots, \tau_{n}(b^n_1, \dots b^i_{k_n})) \wedge\\
	\wedge \bigotimes_{i=1}^n\bigotimes_{j=1}^{k_i} \chi(x^i_j) \leq \chi(x)
	\end{gather*}
	This concludes the proof since every linear term can be written as a concatenation of operations. 

	For affine, relevant and cartesian terms the considerations done in the proof of 
	proposition~\ref{prop:relclosedunderterms} can easily be adapted to the span case.
\end{proof}

\interpretationinducedspan*
\begin{proof}
	$i^*$ is defined as in the relational case, sending algebras of type $(\Sigma_2, E_2)$ to their 
	interpretations of type $(\Sigma_1, E_1)$. Being it identity on morphisms by definition
	(the fact that morphisms of $\linspanvq{(\Sigma_2,E_2)}{Q}$ 
	are also morphisms of $\linspanvq{(\Sigma_1,E_1)}{Q}$ is a direct consequence 
	of proposition~\ref{prop:spanclosedunderterms}) functoriality follows trivially. 
\end{proof}

\logicalfunctorinducedrel*
\begin{proof}
	The proof heavily relies on the fact that logical morphisms preserve models of logical theories: 
	We know that, if $T$ is a logical theory, a logical functor $L: \mathcal{E} \to \mathcal{F}$ 
	preserves every interpretation (that is, every model), of $T$ in $\mathcal{E}$. 
	This is because an interpretation of $T$ in $\mathcal{E}$ assigns to every term 
	and formula its correspondent in the Mitchell-B\'ernabou internal language: 
	Every type is interpreted in a product of objects, and every constant into a morphism of $\mathcal{E}$. 
	The axioms correspond, finally, to commutative diagrams in $\mathcal{E}$. 
	Since these diagrams involve only limits, exponentials and subobject classifiers, 
	they are preserved by $L$ up to isomorphism. This means that the image through $L$ 
	of objects and morphisms that constitute a model of $T$ in $\mathcal{E}$ is a model of $T$ in $\mathcal{F}$. 
	The idea is then to state our definition of composition and identity of 
	$\relvtq{\mathcal{E}}{(\Sigma, E)}{Q}$ in terms of logical theories: In this case the 
	composition and the identity of two algebra-preserving relations will be just a model of this theory in $\mathcal{E}$, 
	and will hence be preserved by $L$. 
	The images through $L$ of our relations will then still satisfy our definition of composition 
	in the internal language of $\mathcal{F}$, guaranteeing that $L(R \circ S) (a,c) = \bigvee_b\{LR(a,b) \otimes LS(b,c)\} = (LR \circ LS)(a,c)$.
	From this, we can define $L^*: \relvtq{\mathcal{E}}{(\Sigma, E)}{Q} \to \relvtq{\mathcal{F}}{(\Sigma, E)}{LQ}$ as follows:
	\begin{itemize}
		\item On objects, $L^*(A) = L(A)$
		\item On morphisms, denoting with $\kappa$ the canonical 
		isomorphism from $LA \times LB$ to $L(A\times B)$, 
		\begin{equation*}
			L^*(R) = LA \times LB \xrightarrow{\kappa} L(A\times B) \xrightarrow{LR} LQ
		\end{equation*}
	\end{itemize}
	Now we have to state what our composition is in terms of logical theories. 
	Given a signature $(\Sigma, E)$, we can define a logical theory 
	\begin{align*}
	T = (&A,B, C, Q,
	\{\sigma^A_i\}_{\sigma_i \in \Sigma}, \{\sigma^B_i\}_{\sigma_i \in \Sigma}, \{\sigma^C_i\}_{\sigma_i \in \Sigma}, \\
	&\otimes, \vee, k, 1_B, R_{AB}, R_{BC}, R_{AC}),
	\end{align*}
	where
	\begin{itemize}
		\item For a given $\sigma_i \in \Sigma$ of ariety $n_i$, 
		\begin{itemize}
			\item $\sigma^A_i$ is a constant of type $A^{A^{n_i}}$
			\item $\sigma^B_i$ is a constant of type $B^{B^{n_i}}$
			\item $\sigma^C_i$ is a constant of type $C^{C^{n_i}}$
		\end{itemize}
		\item $\otimes$ is a constant of type $Q^{Q \times Q}$
		\item $\vee$ is a constant of type $Q^{PQ}$
		\item $k$ is a constant of type $Q$
		\item $1_B$ is a constant of type $Q^{B \times B}$
		\item $R_{AB}, R_{BC}, R_{CD}$ are constants of type 
		$Q^{A \times B}, Q^{B \times C}, Q^{A \times C}$, respectively	
	\end{itemize}
	We require this theory to satisfy the set of axioms 
	$\{ \alpha_{1_B}, \{\alpha_A\}, \{\alpha_B\}, 
	\{\alpha_C\}, \{\alpha_Q\}, \{\alpha_{R_{AB}}\}, \{\alpha_{R_{BC}}\}, 
	\alpha_\textit{comp}\}$, where:
	\begin{itemize}
		\item $\{\alpha_A\}$ is the set of axioms that makes 
		$(A, \{\sigma^A_i\}_{\sigma_i \in \Sigma})$ into an algebra of type $(\Sigma, E)$
		\item $\{\alpha_B\}$ is the set of axioms that makes 
		$(B, \{\sigma^B_i\}_{\sigma_i \in \Sigma})$ into an algebra of type $(\Sigma, E)$
		\item $\{\alpha_C\}$ is the set of axioms that makes 
		$(C, \{\sigma^C_i\}_{\sigma_i \in \Sigma})$ into an algebra of type $(\Sigma, E)$
		\item $\{\alpha_Q\}$ is the set of axioms that makes 
		$(Q, \otimes, k, \bigvee)$ into a quantale
		\item $\alpha_{1_B}$ is the axiom $\forall b,b'. 1_B(b,b') = \bigvee\{k| b = b'\}$
		
		\item $\{\alpha_{R_{AB}}\}$ is the set of all the axioms, one for every $\sigma_i \in\Sigma$ of ariety $n$, of the form
		\begin{gather*}
		\forall_{a_1, \dots, a_n, b_1 \dots, b_n} \bigvee\Bigg\{ \\ 
		R_{AB}(\sigma^A_i(a_1, \dots, a_n),\sigma^B_i(b_1,\dots,b_n)), 
		\bigotimes_{j=1}^n R(a_j,b_j) \Bigg\} =\\ 
		=R_{AB}(\sigma^A_i(a_1, \dots, a_n),\sigma^B_i(b_1,\dots,b_n))
		\end{gather*}
		(Note that in this setting to use the quantale order relation we have to write explicitly what it is. 
		The axiom above is nothing but the algebra preservation axiom for $R_{AB}$ written explicitly using the algebraic lattice structure)
		\item $\{\alpha_{R_{BC}}\}$ is the set of all the axioms, one for every $\sigma_i \in \Sigma$ of ariety $n$,of the form
		\begin{gather*}
		\forall_{b_1 \dots, b_n, c_1, \dots, c_n} \bigvee\Bigg\{ \\
		R_{BC}(\sigma^B_i(b_1, \dots, b_n),\sigma^C_i(c_1,\dots,c_n)), 
		\bigotimes_{j=1}^n R(b_j,c_j) \Bigg\} =\\
		=R_{BC}(\sigma^B_i(b_1, \dots, b_n),\sigma^C_i(c_1,\dots,c_n))
		\end{gather*}
		\item Finally, $\alpha_\textit{comp}$ is the axiom
		\begin{equation*}
			\forall a,c. R_{AC}(a,c) = \bigvee\left\{ R_{AB}(a,b) \otimes R_{BC}(b,c) | b \in B \right\}
		\end{equation*}
	\end{itemize}
	An interpretation of this theory in the topos $\mathcal{E}$ then consists of three morphisms in $\mathcal{E}$
	\begin{equation*}
		R_{AB}:A \times B \to Q \quad R_{BC}:B \times C \to Q \quad R_{AC}:A \times C \to Q
	\end{equation*}
	Where the sets of axioms $\{\alpha_A\}, \{\alpha_B\}, \{\alpha_C\}$ get interpreted 
	into commutative diagrams ensuring that $A,B,C$ are internal algebras of signature $(\Sigma, E)$, respectively, 
	while $\{\alpha_{R_{AB}}\}, \{\alpha_{R_{BC}}\}$ guarantee that $R_{AB}$ and $R_{BC}$ 
	respect the usual algebraic condition. $\{\alpha_Q\}$ gets interpreted into diagrams 
	ensuring that $Q$ is an internal quantale and 
	$\alpha_\textit{comp}$ guarantees that $R_{AC}$ is exactly the composition of 
	relations $R_{AB}, R_{BC}$ in $\relvtq{\mathcal{E}}{(\Sigma, E)}{Q}$.
\end{proof}

\logicalcommutesrel*
\begin{proof}
	Start noting that the graph functor is identity on objects, so trivially $L^*(A)_\circ = L^*A = (L^*A)_\circ$ for every object $A$. For a morphism $f:A \to B$, in $\mathcal{E}$, consider the diagram:
	\begin{equation*}
	\begin{tikzpicture}[scale=0.4, node distance=1.5cm]
	\node (t0) {$LA \times LB$};
	\node[right of=t0, node distance=2.2cm] (t1) {$L(A \times B)$};
	\node[right of=t1, node distance=3cm] (t2) {$L(B \times B)$};
	\node[right of=t2, node distance=2.4cm] (t3) {$L(Q)$};
	\node[below of=t1] (b1) {$LA \times LB$};
	\node[below of=t2] (b2) {$LB \times LB$};	
	\draw (t0) to node[above]{iso} (t1);
	\draw[double distance=1pt] (t0) to (b1);
	\draw[->] (t1) to node[above]{$L(f \times 1_B)$} (t2);
	\draw[->] (b1) to node[below]{$Lf \times L 1_B$} (b2);
	\draw[->] (t2) to node[above]{$L(\text{id}^\textbf{Rel}_B)$} (t3);
	\draw[->] (b2) to node[below right]{$\text{id}^\textbf{Rel}_{LB}$} (t3);
	\draw (t1) to node[right]{iso} (b1);
	\draw (t2) to node[left]{iso} (b2);
	\end{tikzpicture}
	\end{equation*}
	Where $\text{id}^\textbf{Rel}_B$ is the morphism of $\mathcal{E}$ that defines $1_B$ in $\relvtq{\mathcal{E}}{\Sigma, E}{Q}$.
	The top row of the diagram is just $L^*(f)_\circ$, while the bottom one is $(L^*f)_\circ$. The left triangle commutes trivially, the center square commutes because $L$ preserves products, the right triangle commutes because $L$ preserves relational identities (previous proposition). Preservation of the converse follows trivially from the fact that any logical functor preserves products.
\end{proof}

\logicalfunctorinducedspan*
\begin{proof}
	Here the same considerations used to prove theorem~\ref{thm:logicalfunctorinducedrel} hold. 
	Given a signature $(\Sigma, E)$, the logical theory we use is
	\[
	T = (X, A, B, Q,\{\sigma^A_i\}_{\sigma_i \in \Sigma}, \{\sigma^B_i\}_{\sigma_i \in \Sigma}, f,g,\chi, \otimes, \leq, k),
	\]
	where
	\begin{itemize}
	\item For a given $\sigma_i \in \Sigma$ of ariety $n_i$, 
	\begin{itemize}
		\item $\sigma^A_i$ is a constant of type $A^{A^{n_i}}$
		\item $\sigma^B_i$ is a constant of type $B^{B^{n_i}}$
		\item $\sigma^C_i$ is a constant of type $C^{C^{n_i}}$
	\end{itemize}
	\item $f,g,\chi$ are constants of type $X^A, X^B, X^Q"$, respectively
	\item $\otimes$ is a constant of type $Q^{Q \times Q}$
	\item $\leq$ is a constant of type $\Omega^{Q \times Q}$
	\item $k$ is a constant of type $Q$
	\end{itemize}
	We require this theory to satisfy the set of axioms $\{ \{\alpha_A\}, \{\alpha_B\}, \{\alpha_Q\}, \{\alpha_{X}\}\}$, where:
	\begin{itemize}
		\item $\{\alpha_A\}$ is the set of axioms that makes 
		$(A, \{\sigma^A_i\}_{\sigma_i \in \Sigma})$ into an algebra of type $(\Sigma, E)$
	\item $\{\alpha_B\}$ is the set of axioms that makes 
	$(B, \{\sigma^B_i\}_{\sigma_i \in \Sigma})$ into an algebra of type $(\Sigma, E)$
	\item $\{\alpha_C\}$ is the set of axioms that makes 
	$(C, \{\sigma^C_i\}_{\sigma_i \in \Sigma})$ into an algebra of type $(\Sigma, E)$
	\item $\{\alpha_Q\}$ is the set of axioms that makes 
	$(Q, \otimes, k, \leq)$ into a partially ordered monoid
	\item $\{\alpha_{X}\}$ is the set of axioms, one for every 
	$\sigma \in\Sigma$ of ariety $n$, of the form
		\begin{gather*}
		\forall_{x_1, \dots, x_n}\exists x. f(x) = \sigma^A_i (f(x_1), \dots, f(x_n) \wedge \\
		\wedge \left( g(x) = \sigma^B_i (g(x_1), \dots, g(x_n) \right) \wedge \\
		\wedge \bigotimes_{j= 1}^n \chi(x_j) \leq \chi(x)
		\end{gather*}
	\end{itemize}
	A model of $T$ in $\mathcal{E}$ is just a span that respects the algebraic structure and we know that $L$ preserves this condition. 
	$L^*$ then agrees with $L$ on objects and is defined as $(LX, Lf,Lg, L\chi)$ on the morphism $(X, f,g, \chi)$.
	For composition and identity we do not need to invoke any logical theory: The identity span is of the form $(X,1_X,1_X,\chi_k)$, 
	and the span part is clearly preserved because functors preserve identities in general. 
	The quantale part $\chi_k$ is the morphism $A \to 1 \to Q$ where the final arrow sends the terminal object to the quantale unit. 
	Again, being $L$ logical this is trivially preserved. 
	For composition, note that the span part is composed via pullbacks and $L$ preserves limits. 
	For the quantale part we have, supposing $(Z, h,k, \zeta)$ to be the composite of $(X, f,g, \chi)$ and $(Y, f',g',\upsilon)$,
	\begin{equation*}
	\begin{tikzpicture}[scale=0.4, node distance=1.5cm, ->]
	\node (t0) {$LZ$};
	\node[right of=t0, node distance=2.8cm] (t1) {$L(X \times Y)$};
	\node[right of=t1, node distance=3cm] (t2) {$L(Q \times Q)$};
	\node[right of=t2, node distance=2cm] (t3) {$L(Q)$};
	\node[below of=t1] (b1) {$LX \times LY$};
	\node[below of=t2] (b2) {$LQ \times LQ$};	
	\draw (t0) to node[above]{$L\!\!<\!\!p_1,p_2\!\!>$} (t1);
	\draw (t0) to node[below left]{$<\!\!Lp_1,Lp_2\!\!>$} (b1);
	\draw (t1) to node[above]{$L(\chi \times \upsilon)$} (t2);
	\draw (b1) to node[below]{$L\chi \times L\upsilon$} (b2);
	\draw (t2) to node[above]{$L(\otimes)$} (t3);
	\draw (b2) to node[below]{$\otimes$} (t3);
	\draw (t1) to node[right]{iso} (b1);
	\draw (t2) to node[left]{iso} (b2);
	\end{tikzpicture}
	\end{equation*}
	Where $p_1, p_2$ are the pullback projections.
	The top row is the image of $\zeta$ through $L$. 
	The triangle on the left and the square on the center commute because $L$ preserves limits, 
	while the triangle on the right commutes because every partially ordered monoid is obviously a model of a theory, 
	so the multiplication of $Q$ gets carried in the multiplication of $LQ$.
\end{proof}

\thebox*
\begin{proof}
	Here the notation $A \simeq B$ will denote that $A$ and $B$ are isomorphic.
	$i^*$ trivially commutes with $h^*$, since the first is identity on morphisms 
	and the second is identity on objects; for the very same reason, $i^*$ commutes with $V$.
	$V$ commutes with $h^*$ because the former acts by postcomposition with a homomorphism of quantales, that commutes with joins
	and orders. 
	
	To show that $L^*i^* \simeq i^*L^*$, note that for morphisms this is trivial, being $i^*$ the identity on them. 
	Let then $\langle\langle{A, \sigma^A_j}\rangle\rangle$ be an object of, say, 
	$\relvtq{\mathcal{E}}{(\Sigma_2, E_2)}{Q}$, and consider 
	$L^*i^*\langle\langle{A, \sigma^A_j}\rangle\rangle$. By definition this is equal to 
	$L^*\langle\langle{A, i(\sigma^A)_{j'}}\rangle\rangle$, 
	where every $i(\sigma^A)_{j'}$ is a term derived from the $\sigma^A_j$, 
	so a composition of $\sigma^A_j$ (and eventually diagonals and projections, 
	depending on the interpretation). Being $L$ logical, operations of $A$ 
	get carried into operations of $LA$, hence 
	$L^*\langle\langle{A, i(\sigma^A)_{j'}}\rangle\rangle = \langle\langle{LA, Li(\sigma^A)_{j'}}\rangle\rangle$ 
	is an algebra of type $(\Sigma_1, E_1)$ in $\relvtq{\mathcal{F}}{(\Sigma_1, E_1)}{LQ}$.
	But, being $i(\sigma^A)_{j'}$ a composition of operations, projections and diagonals, 
	and being $L$ product preserving, it is $Li(\sigma^A)_{j'} \simeq i(\sigma^{LA})_{j'}$. Hence
	\begin{align*}
		L^*i^*\langle\langle{A, \sigma^A_j}\rangle\rangle &= \langle\langle{LA, Li(\sigma^A)_{j'} }\rangle\rangle\\
		&\simeq \langle\langle{LA, i(\sigma^{LA})_{j'}}\rangle\rangle \\
		&= i^*L^*\langle\langle{A, \sigma^A_j}\rangle\rangle
	\end{align*}
	The proof is the same when $L^*$ and $i^*$ act on spans.

	To prove that $L^*V \simeq V^*L^*$, consider the following logical theory:
	\begin{equation*}
		T = (X, A, B, Q,\{\sigma^A_i\}_{\sigma_i \in \Sigma}, \{\sigma^B_i\}_{\sigma_i \in \Sigma}, f,g, \chi, \otimes, \vee, k, R)
	\end{equation*}	
	Where:
	\begin{itemize}
	\item $(X, A, B, Q,\{\sigma^A_i\}_{\sigma_i \in \Sigma}, \{\sigma^B_i\}_{\sigma_i \in \Sigma}, f,g, \chi, \otimes, \vee, k)$ 
	is the fragment that states that
	$(Q, \otimes, k, \bigvee)$ is a quantale, 
	that $\langle\langle A,\{\sigma^A_i\}_{\sigma_i \in \Sigma} \rangle\rangle$ and 
	$\langle\langle B, \{\sigma^B_i\}_{\sigma_i \in \Sigma}\rangle\rangle$ are algebras 
	of the required signature and that $(X,f,g,\chi)$ is an algebraic preserving span over $Q$, 
	with all the obvious axioms required to hold 
	(see proof of theorems~\ref{thm:logicalfunctorinducedrel} and~\ref{thm:logicalfunctorinducedspan} for details)
	\item $R$ is a constant of type $Q^{A \times B}$ together with the axioms 
	that say it is an algebraic preserving relation over $Q$ 
	(again refer to the relational case in theorem~\ref{thm:logicalfunctorinducedrel})
	\item The additional axiom
	\begin{equation*}
		R(a,b) = \bigvee \left\{\chi(m)| \exists m. (s(m)=a \wedge t(m) = b)\right\}
	\end{equation*}
	is satisfied.
	\end{itemize}
	This logical theory expresses exactly the fact that $R$ is a relation coming from a span in the sense of the order functor, 
	so if $R = V(X,s,t,\chi)$ then $R$ and $(X,s,t,\chi)$ are a model for $T$. 
	From this we get an isomorphism between $LR$ and $V(LX,Lf,Lg, L\chi)$, and hence
	\begin{equation*}
		L^*V(X,f,g,\chi) = L^*R \simeq LR \simeq V(LX,Lf,Lg, L\chi)
	\end{equation*}
	
	Finally, to verify that $L^*h^* \simeq h^*L^*$, just note that it is possible to state what a quantale homomorphism is in terms of logical theories. This guarantees that if $h:Q_1 \to Q_2$ is a homomorphism of quantales, so is $Lh$. Everything then follows from the fact that $h^*$ acts by postcomposition and $L$ respects it.
\end{proof}

\end{document}